\documentclass[a4paper,accepted=2025-08-08,onecolumn]{quantumarticle}
\pdfoutput=1
\usepackage{authblk}
\usepackage{geometry}
\usepackage[utf8]{inputenc}
\usepackage[english]{babel}
\usepackage[T1]{fontenc}
\usepackage{amsmath}
\usepackage{amsthm}
\newtheorem{remark}{Remark}
\usepackage{xcolor}
\colorlet{myPurple}{blue!40!red}
\colorlet{myPurplee}{blue!10!red}
\colorlet{myCyan}{cyan!60!gray}
\colorlet{myRed}{red!66!black}
\usepackage{tikz}
\usepackage{pgfplots}
\pgfplotsset{compat=1.14}
\usepackage{hyperref}
\usepackage[normalem]{ulem}
\usepackage{exscale}
\usepackage{graphicx}
\usepackage{amsmath}
\usepackage{latexsym}
\usepackage{amsfonts}
\usepackage{amssymb}
\usepackage[T1]{fontenc}
\usepackage{amsthm}
\usepackage{enumerate}
\usepackage{bbold}
\usepackage{color}
\usepackage{nicefrac}%
\usepackage{cleveref}

\newcommand{\sket}[1]{{\ensuremath{\lvert#1\rangle}}}
\newcommand{\lket}[1]{{\ensuremath{\left\lvert#1\right\rangle}}}
\newcommand{\ket}[1]{\if@display\lket{#1}\else\sket{#1}\fi}
\newcommand{\tp}{\otimes}
\newcommand{\sbra}[1]{{\ensuremath{\langle#1\rvert}}}
\newcommand{\lbra}[1]{{\ensuremath{\left\langle#1\right\rvert}}}
\newcommand{\bra}[1]{\if@display\lbra{#1}\else\sbra{#1}\fi}

\newcommand{\sbraket}[2]{{\ensuremath{\langle#1\rvert#2\rangle}}}
\newcommand{\lbraket}[2]{{\ensuremath{\left\langle#1\!\left\rvert\vphantom{#1}#2\right.\!\right\rangle}}}
\newcommand{\braket}[2]{\if@display\lbraket{#1}{#2}\else\sbraket{#1}{#2}\fi}

\newcommand{\sketbra}[2]{{\ensuremath{\lvert #1\rangle\!\langle #2\rvert}}}
\newcommand{\lketbra}[2]{{\ensuremath{\left\lvert #1\right\rangle\!\!\left\langle #2\right\rvert}}}
\newcommand{\ketbra}[2]{\if@display\lketbra{#1}{#2}\else\sketbra{#1}{#2}\fi}
\usepackage{tcolorbox}
\usepackage{mathtools}

\newcommand{\proj}[1]{\ketbra{#1}{#1}}

\newcommand{\physstatedm}{\varrho_{\rA\rB}}

\newcommand{\iso}{\Phi}
\newcommand{\junk}{\ket{\xi}}

\newcommand{\tr}{\mathrm{Tr}}

\newcommand{\idd}{\mathds{1}}

\newcommand{\rA}{\text{A}}
\newcommand{\rB}{\text{B}}

\newcommand{\rP}{\text{P}}
\newcommand{\M}{\mathsf{M}}

\newcommand{\X}{\mathsf{X}}
\newcommand{\x}{\left|\mathsf{X}\right|}
\newcommand{\Y}{\mathsf{Y}}
\newcommand{\y}{\left|\mathsf{Y}\right|}
\newcommand{\A}{\mathsf{A}}
\newcommand{\B}{\mathsf{B}}

\newcommand{\xe}{\mathtt{x}}
\newcommand{\aen}{\mathtt{a}}
\newcommand{\Exp}{\mathop{\mathds{E}}}

\usepackage{tikz}
\usepackage{lipsum}
\usepackage{mathrsfs}
\theoremstyle{plain}
\newtheorem{thm}{Theorem}%
\newtheorem{thm*}{Result (informal)}
\newtheorem{lem}{Lemma}

\newtheorem{corollary}{Corollary}

\newtheorem{defi}{Definition}
\usepackage{graphicx}
\usepackage{bm}
\usepackage{dsfont}
\usepackage{tikz}
\usepackage[T1]{fontenc}
\usepackage{amsthm}
\usepackage{array}
\usepackage{amssymb}
\usepackage{amsfonts}
\usepackage{cancel}
\usepackage[toc,page]{appendix}
\usepackage{multirow}
\usepackage{color}
\usepackage{calrsfs}
\usetikzlibrary{backgrounds,decorations.pathreplacing,calc}

\usepackage{tkz-euclide}

\newcommand{\gen}{\mathsf{Gen}}
\newcommand{\eval}{\mathsf{Eval}}
\newcommand{\enc}{\mathsf{Enc}}
\newcommand{\dec}{\mathsf{Dec}}
\newcommand{\sk}{\mathsf{sk}}
\newcommand{\ct}{\mathsf{ct}}

\DeclareMathAlphabet{\mathcal}{OMS}{cmsy}{m}{n}
\begin{document}
\title{Quantum bounds for compiled XOR games and $d$-outcome CHSH games}

\author{Matilde Baroni}
\affiliation{Sorbonne Université, CNRS, LIP6, 4 place Jussieu, 75005 Paris, France}
\email{matilde.baroni@lip6.fr}
\author{Quoc-Huy Vu}
\affiliation{De Vinci Higher Education, De Vinci Research Center, Paris, France}
\author{Boris Bourdoncle}
\affiliation{Quandela, 7 Rue L\'{e}onard de Vinci, 91300 Massy, France}
\author{Eleni Diamanti}
\affiliation{Sorbonne Université, CNRS, LIP6, 4 place Jussieu, 75005 Paris, France}
\author{Damian Markham}
\affiliation{Sorbonne Université, CNRS, LIP6, 4 place Jussieu, 75005 Paris, France}
\author{Ivan \v{S}upi\'{c}}
\affiliation{Sorbonne Université, CNRS, LIP6, 4 place Jussieu, 75005 Paris, France}

\begin{abstract}
    Nonlocal games play a crucial role in quantum information theory and have numerous applications in certification and cryptographic protocols. Kalai et al. (STOC 2023) introduced a procedure to compile a nonlocal game into a single-prover interactive proof, using a quantum homomorphic encryption scheme, and showed that their compilation method preserves the classical bound of the game. Natarajan and Zhang (FOCS 2023) then showed that the quantum bound is preserved for the specific case of the CHSH game. Extending the proof techniques of Natarajan and Zhang, we show that the compilation procedure of Kalai et al. preserves the quantum bound for two classes of games: bipartite XOR games and a higher dimensional generalization of CHSH (the SATWAP inequality), that we refer to as $d$-outcome CHSH games.
    We also establish that, for any pair of qubit measurements, there exists a compiled XOR game such that its near-optimal winning probability serves as a robust self-test for that particular pair of measurements.
    Finally, we derive computational self-testing of three anticommuting qubit observables, based on the compilation of the nonlocal game corresponding to the so-called elegant Bell inequality. 
\end{abstract}

\section{Introduction}\label{intro}

In quantum information theory, nonlocal games are a class of games that exemplifies the separation between classical and quantum resources. In such games, a referee sends classical inputs to two or more distant parties, and the parties reply with classical outputs. A predicate on the tuples of inputs and outputs defines the winning conditions for the game. In some cases, if the distant parties have access to quantum resources, namely entangled states, they can achieve a score higher than if they only have classical resources. Building on the work of  Bell~\cite{Bell1964Einstein}, Clauser, Horne, Shimony and Holt introduced the setting for the most famous example of such a game, the CHSH game~\cite{CHSH1969}. While players with classical resources can only reach a maximal value of 0.75, players sharing entanglement can obtain a value of $\cos^2(\pi/8)$. 

In addition to its far-reaching foundational interest, Bell nonlocality~\cite{BCP2014Bell} enjoys a rich array of applications concerned with the certification of quantum resources and cryptographic tasks, such as device-independent quantum key distribution~\cite{ABG2007Device, VV2014Fully, ADF2018Practical}, certified randomness~\cite{Colbeck2007Quantum, PAM2010Random, AM2016Certified}, self-testing~\cite{MY2004Self, SB2020Self} and protocols for verifiable blind delegated quantum computing~\cite{reichardt2013classical,coladangelo2019verifier}. All these protocols are characterized by information-theoretic security. Bell nonlocality is also connected to the powerful notion of multiprover interactive proof systems in quantum complexity theory, $\textrm{MIP}^*$, in which two or more non-communicating provers are allowed to share quantum entanglement. In the classical setting, techniques and results from studying MIP (the classical version of $\textrm{MIP}^*$) usually also find applications in the cryptographic setting, where all parties are computationally bounded, such as the celebrated PCP theorem~\cite{babai1991non} and succinct arguments~\cite{STOC:KalRazRot13,FOCS:MetNatZha24}.

The framework of Bell nonlocality however requires at least two non-communicating parties, and the strongest guarantees of information-theoretic security can thus only be achieved if these parties are spatially separated, which is experimentally very challenging.  While effective for foundational proofs of quantumness, applying such setups to computing platforms faces challenges due to their incompatibility with Bell-type scenarios. From a pragmatical point of view, if one wants to certify resources on a quantum computer as it is an integral device, spatial separation and the absence of communication cannot be imposed. Thus, a possibility to translate certification tools requiring two or more spatially separated devices to single device setting would significantly boost their applicability.
Interestingly, in the other area where Bell nonlocality finds its application, namely quantum computational complexity that we commented above, there has been a progress in moving from two non-communicating provers to a single one. Crucially, a series of works~\cite{aiello2000fast,STOC:KalRazRot14} showed that, using a homomorphic encryption scheme, any MIP sound against non-signalling provers can be compiled into a single-prover protocol in which the single prover is computationally bounded and restrained by cryptographic tools.

Cryptography seems to offer a very promising tool: (quantum) homomorphic encryption, which allows to perform computations on encrypted (quantum) data without decrypting them. 
This could be used to mimic spatial separation by hiding to the single prover the complete knowledge of the inputs, which would otherwise make the game trivial. 
Kalai, Lombardi, Vaikuntanathan and Yang formalised in~\cite{STOC:KLVY23} the idea of using quantum homomorphic encryption~\cite{FOCS:Mahadev18b,C:Brakerski18} to emulate spatial separation between the nonlocal parties.
They propose a procedure (see~\Cref{fig:klvy}) to compile every $k$-players nonlocal game into a single-prover interactive proof of $2k$-rounds, by fixing a sequential structure of the input-output requests and encrypting the first $k-1$ rounds.
They then prove: (i) the \emph{completeness} of the procedure, that is, there is an explicit and efficient quantum strategy that achieves the quantum bound of the original nonlocal game; and (ii) its \emph{classical soundness}, that is, no classical prover can outperform the optimal classical bound up to a negligible function in the security parameter, which is the advantage of the classical adversary in the indistinguishability (IND-CPA) security game of the encryption scheme.

\begin{figure}[htbp]
    \begin{minipage}[t]{0.6\linewidth}
        \centering
\begin{tikzpicture}[scale=0.9]
\draw[->,>=stealth] (-0.5,4) -- (-0.5,0) node[midway,left] {time};

\draw (0,1) rectangle (2,3) node[midway] (alice) {Alice};

\draw (3,0) rectangle (5,4) node[midway] (nl_verifier) {Verifier};

\draw (6,1) rectangle (8,3) node[midway] (bob) {Bob};

\draw[->,>=stealth] (3,2.5) -- (2,2.5) node[midway,above] {$x$};
\draw[->,>=stealth] (2,1.5) -- (3,1.5) node[midway,above] {$a$};

\draw[->,>=stealth] (5,2.5) -- (6,2.5) node[midway,above] {$y$};
\draw[->,>=stealth] (6,1.5) -- (5,1.5) node[midway,above] {$b$};
\end{tikzpicture}         \label{fig:2provers}
    \end{minipage}%
    \begin{minipage}[t]{0.4\linewidth}
        \centering
\begin{tikzpicture}[scale=0.9]%
\draw (0,0) rectangle (2,4) node[midway] (verifier) {Verifier};

\draw (4.5,0) rectangle (6.5,4) node[midway,text width=2cm,align=center] (singleprover) {Single \\ Prover};

\draw[->,>=stealth] (2,3.5) -- (4.5,3.5) node[midway,above] {$\xe=\enc(x)$};
\draw[->,>=stealth] (4.5,2.5) -- (2,2.5) node[midway,above] {$\aen=\enc(a)$};

\draw[->,>=stealth] (2,1.5) -- (4.5,1.5) node[midway,above] {$y$};
\draw[->,>=stealth] (4.5,0.5) -- (2,0.5) node[midway,above] {$b$};

\end{tikzpicture}         \label{fig:singleprover}
    \end{minipage}
    \caption{Pictorial representation of the Kalai \emph{et. al.} compilation protocol for $2$-player nonlocal games. On the left, a general $2$-player nonlocal game; the two parties are spatially separated, and only communicate to a classical verifier which is sampling questions $(x,y)$ and collecting their answers $(a,b)$. On the right, the single prover game resulting from the compilation procedure. In this representation, time flows downwards.}
    \label{fig:klvy}
\end{figure}
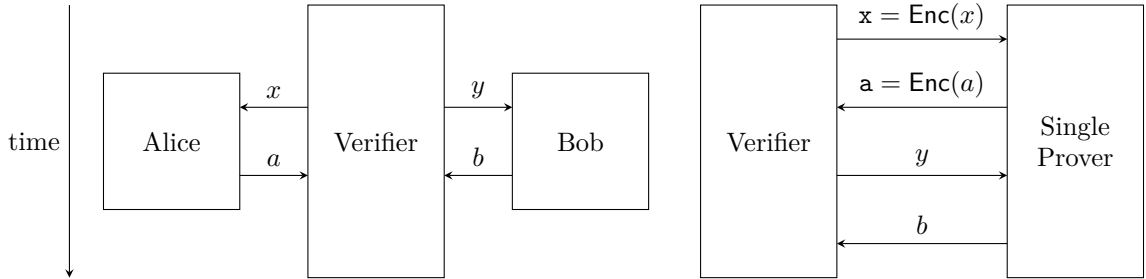

However \cite{STOC:KLVY23} does not provide an upper bound to the optimal score of a quantum prover.
Indeed they use classical rewinding techniques to show the soundness against classical adversaries, which is no longer possible against quantum provers.
Follow-up works by Natarajan and Zhang~\cite{FOCS:NatZha23} and Brakerski \emph{et al.}~\cite{C:BGKPV23} prove that the quantum bound is preserved in the specific case of the compiled CHSH game.
Their proofs heavily relies on specific properties of CHSH, providing little insight on whether Kalai \emph{et al.} compiler preserves the quantum bound for any other nonlocal games.

\subsection{Our Contributions}
In this work, we continue this research direction and provide a partial answer for the above open question. In particular, we prove that the quantum bound is preserved for two classes of compiled games: XOR games, where the predicate depends on the logical XOR of the outputs, and $d$-outcome CHSH games, that generalise the CHSH game to a scenario involving many-output measurements.
We also extend our results to self-testing, which relies on the fact that for specific non-local games there is a unique strategy that can achieve the optimal quantum score; then, reaching this bound is a device-independent certification of quantum resources. 

The techniques we use are based on the cryptographic SOS decomposition delineated in \cite{FOCS:NatZha23}.
The score of a nonlocal game is usually formulated as the expectation value of some function of the observables of the players. This is closely related to the Bell operator $\mathcal{B}$, and its shifted version $\beta \mathds{1} - \mathcal{B}$, where $\beta$ is what we call the shift.
SOS decomposition - standing for sum of squares - is a method widely used in Bell nonlocality to prove quantum bounds. It consists on mathematically rearranging the terms of the shifted Bell operator as a sum of polynomials squared, whose expectation value is positive semi-definite by definition. 
\begin{equation*}
   \bra{\psi} \beta \mathds{1} - \mathcal{B} \ket{\psi}= \sum_i \bra{\psi} P_i^2 \ket{\psi} \geq 0 \implies 
    \bra{\psi}\mathcal{B}\ket{\psi} \leq \beta
\end{equation*}
Hence the shifted Bell operator is also positive semi-definite, meaning that the shift $\beta$ is an upper bound for the optimal quantum score. The bound becomes tight if an explicit quantum strategy achieving $\beta$ is found.
If the polynomials appearing in the decomposition are of degree $n$, then we say we have a $n$th-order SOS decomposition.

We consider the the pseudo-expectation value map $\Tilde{\mathds{E}}$ defined in~\cite{FOCS:NatZha23}, that maps polynomials of nonlocal observables to correlations observed in the compiled nonlocal game.
This function allows to translate the SOS decomposition proof from the nonlocal to the compiled game.
For polynomials of observables with a physical interpretation, the map corresponds to a simple change of variables, compatible with the structure of the compiled game; therefore it acts trivially on the identity, and when applied to the Bell functional it gives the score of the compiled game.
The name pseudo-expectation refers to the fact that it is not a positive semi-definite function by definition. 
In particular, the pseudo-expectation of the polynomials squared from the SOS decomposition is not positive by definition anymore. 
Nevertheless, elaborating on the results of~\cite{FOCS:NatZha23}, we prove that - for specific polynomials -  these quantities are still positive up to a negligible function in the security parameter $\kappa$ of the QHE scheme.

\begin{thm*}
    Consider a Bell inequality $\mathcal{B}$ with a SOS decomposition, whose polynomials can be written as $P_i = A_i - \hat{B}_i$, where $\{A_i\}_i$ are Alice's observables and $\{\hat{B}_i\}_i$ are linear sums of Bob's observables. Then there exists a negligible function $\textrm{negl}(\cdot)$ such that $\Tilde{\mathds{E}}\left[P_i^2
    \right] \geq -\textrm{negl}(\kappa)$.
\end{thm*}
Technically, this is done by extending cryptographic arguments based on the security of the encryption scheme. Then, applying the pseudo-expectation map to the nonlocal SOS decomposition, we prove that the quantum bound of the compiled game is preserved up to a negligible function in the security parameter. 
\begin{equation*}
     \Tilde{\mathds{E}}[\beta \mathds{1} - \mathcal{B}] = \sum_i \Tilde{\mathds{E}}[P_i^2] \geq - \textrm{negl}(\kappa) \implies  \Tilde{\mathds{E}}[\mathcal{B}] \leq \beta + \textrm{negl}(\kappa)
\end{equation*}

The pseudo-expectation map is well defined for polynomials of degree $2$ in the number of observables, restricting us to work with first-order SOS decompositions.
We lack of a meaningful interpretation for polynomials containing two different Alice's observables in the compiled scenario, but we can fix the shape of the SOS polynomials $P_i$ to contain only one observable for Alice for several non-local games of interest.

In particular, in the paper we revise some known results for XOR games, and derive a new result on the decomposition of the shifted game operator. 
Notably, we prove that it is always possible to have a decomposition with polynomials $P_i$ containing only one Alice's observable. This helps us to prove that the quantum score of all XOR games is preserved by Kalai \emph{et al.} compilation.

\begin{thm*}
    Given a bipartite XOR game with optimal quantum winning probability $\omega_q$, the optimal quantum winning probability of the compiled XOR game is $\omega_q+\textrm{negl}(\kappa)$.
\end{thm*}

We also extend the pseudo-expectation map to treat higher dimensional inputs and outputs, hence generalised observables. As an example we apply it to the SATWAP inequality proposed in \cite{SATWAP17}, a generalisation of CHSH for which we know there exists a SOS decomposition in our desired form \cite{SSKA21}.
We prove that the quantum score of the compiled game is preserved in this case as well.
Apart from its mathematical relevance, this could lead to much more efficient certifications.

\begin{thm*}
    Consider the $d$-dimensional SATWAP Bell inequality $\mathcal{B}_d$ with quantum bound $\omega_q(\mathcal{B}_d)$, and its compiled version.
    If $d$ is polynomial w.r.t. the security parameter $\kappa$, then the quantum bound of the compiled SATWAP inequality is $\omega_q(\mathcal{B}_d) +\textrm{negl}(\kappa)$.
\end{thm*}

The SOS decomposition imposes constraints on the strategies achieving the optimal quantum score; sometimes these constraints are enough to uniquely identify these strategies, leading to self-testing protocols.
From the cryptographic SOS decomposition we develop computational self-testing techniques for interactive single prover games, built on top of the nonlocal self-testing proofs.

We show that for any pair of single-qubit observables, there is a tailored XOR game whose compiled version forces the prover’s device to implement (up to isometry) that specific pair.

\begin{thm*}
    For every pair of qubit observables , there exists an  XOR game $G$ such that any polynomial-time prover in the compiled protocol $G'$  who wins with probability at least $\omega_q(G) - \varepsilon$ must implement measurement operators within $O(\sqrt{\varepsilon}, \textrm{negl}(\kappa))$, up to a local isometry.
\end{thm*}

Our construction begins with the well-known information-theoretic self-test of two arbitrary qubit measurements via a suitable Bell inequality, then embeds it into the KLVY framework. A careful soundness analysis shows that the usual robust‐self-testing error bounds degrade only by a negligible factor in $\kappa$. To our knowledge, this is the first computational self-test of arbitrary two qubit measurement observables in a single-prover setting.

Beyond pairs of observables, practical quantum certification often demands a full Bloch‐sphere basis. We give the first compiled protocol that self-tests three anti-commuting qubit observables, equivalently, the Pauli set ${\sigma_x,\sigma_y,\sigma_z}$
through interaction with a single prover.

\begin{thm*}
    There is a single-prover compiled protocol, based on the elegant Bell inequality, such that any efficient prover achieving winning probability within $\delta$ of the quantum optimum must implement three anti-commuting Pauli measurements ${\sigma_x,\sigma_y,\sigma_z}$ up to $O(\delta,\textrm{negl}(\kappa))$ error and local isometry.
\end{thm*}

Our analysis combines the elegant inequality’s tight self-testing features with the general compilation machinery. The result provides a complete toolkit for certifying an arbitrary single-qubit state or operation under standard cryptographic assumptions.

\subsection{Outline of the paper}
The paper is structured in the following way.
Section \ref{sec:preliminaries} is dedicated to preliminary notions.
We introduce nonlocal games and Bell inequalities, with a focus on a sub-class: XOR games. For these, in Section~\ref{sec:xor} we revise some known results and derive a new result on the decomposition of the shifted game operator. Then, we introduce quantum homomorphic encryption and Kalai's compilation protocol \cite{STOC:KLVY23}.
We elaborate on the results of \cite{FOCS:NatZha23}, in particular on their idea of defining a pseudo-expectation value that maps polynomials of measurement observables to correlations observed in the compiled nonlocal game. 

The second part focuses on those cases for which we can prove that such a pseudo-expectation map can be used to upper bound the optimal quantum winning probabilities in compiled nonlocal games.
In Section~\ref{sec:qb-xor} we use the decomposition introduced in the preliminaries to prove the soundness of the quantum bound for any compiled XOR game.
Section~\ref{sec:st_qudits} generalizes our approach to encompass an inequality with a higher number of outputs.
Finally, in section~\ref{sec:st_2qubitmeas} we develop computational self-testing techniques for interactive single prover games, built on top of the nonlocal self-testing proofs. In particular, we show that a certain class of compiled XOR games can be used to self-test any pair of qubit measurements. Furthermore we introduce a compiled game which can self-test all three anti-commuting qubit observables. 

\subsection{Concurrent and independent work}

While finishing this manuscript we became aware of a recent work by Cui \textit{et al.}~\cite{CMMN2024Computational} who also show that  the compilation procedure of Kalai et al.~\cite{STOC:KLVY23} preserves the quantum bias of all XOR games. 
Even though both papers use a similar techniques based on~\cite{FOCS:NatZha23}, \cite{CMMN2024Computational} and our work also achieve different results. Their work shows a type of self-testing result for all compiled XOR games, while we explore the self-testing properties of carefully designed compiled XOR games, such that any pair of qubit measurements and tomographically complete sets of qubit measurements can be self-tested from the optimal winning probability. The work of  Cui \textit{et al.} explores the compilation of parallelly repeated XOR games and the Magic Square game, while we work with the compiled $d$-CHSH game which cannot be seen as a parallel repetition of XOR games.

\subsection{Subsequent works}
\label{sect:subsequent}
For completeness, we also highlight subsequent works that advanced our understanding on the quantum soundness of compiled games.

In \cite{mehta2025selftestingcompiledsettingtiltedchsh} the authors use similar techniques to show that the quantum bound of the tilted CHSH inequality is preserved under the compilation; this was the first example of an inequality not maximized by a maximally entangled state. 
Furthermore, they contribute in formalizing which kind of self-testing statements can be potentially derived from compiled games.

In \cite{kulpe2024boundquantumvaluecompiled}, a new algebraic framework is introduced, leading to show that for all bipartite games the compiled quantum value is asymptotically (with respect to the security of the cryptographic scheme) bounded by the quantum commuting operator value. This elegant general result provides a new point-of-view on compiled games; however it is an asymptotic result, hence it cannot be directly applied to practical scenarios, because it only holds in the case of perfect security.
Many interesting questions are still open.
How can we make this result more practical, i.e. when is it possible to estimate the value of the compiled game for a finite level of security and the speed of convergence? How to extend the same techniques to games for more than $2$ players?

\section{Preliminaries}\label{sec:preliminaries}

\subsection{Nonlocal games and Bell inequalities}\label{NlBI}

In a two-party nonlocal game, a referee randomly selects questions, also called inputs, according to a predetermined distribution and sends them to the non-communicating parties, usually called Alice and Bob. Upon receiving the answers, also called outputs, from Alice and Bob, the referee determines whether they won or lost based on the game rule, which is public. Winning is nontrivial because the players can't communicate during the game: they must generate their output solely based on their respective inputs and potentially a shared resource. $x \in \X$ and $y \in \Y$ denote the questions sent to Alice and Bob, respectively, $a \in \A$ and $b \in \B$ denote their respective answers. $\left|\cdot\right|$ denotes the cardinality of a set. The referee samples the questions from a distribution $q(x,y)$, and the game rule is a predicate $V : (a,b,x,y) \mapsto \{0,1\}$. Alice and Bob define a strategy to answer the questions, based on the resources that they have, and the conditional probabilities of obtaining outputs $a$ and $b$ when inputting $x$ and $y$, is denoted by $p(a,b|x,y)$. The winning probability of the game is then given by:

\begin{equation}
    \omega = \sum_{x,y}q(x,y)V(a,b,x,y)p(a,b|x,y).
\end{equation}

If Alice and Bob use classical resources, in the most general case they can share a random variable $t$ sampled from a distribution $p(t)$, and they can determine their respective output based on the value of $t$ and the received input. In the context of Bell nonlocality, such strategy corresponds to a local hidden variable model (LHV). In that case, the maximal winning probability is called the classical score of the game. Alice and Bob can always reach the maximal winning probability with deterministic response functions $f(x,t)$, $g(y,t)$ to choose their outputs~\cite{Fine1982Hidden}, and the maximal classical score is then equal to:
\begin{equation}
    \omega_c = \max_{f,g}\sum_{x,y,t}p(t)q(x,y)V(f(x,t),g(y,t),x,y).
\end{equation}

If Alice and Bob have access to quantum resources, they can share a quantum state $\ket{\psi}$ and perform quantum measurements on it. For each output, they can perform a different measurement, meaning that Alice has $\left|\X\right|$ different measurements, one for each for each $x$, characterized by $\left|\A\right|$ operators $\{M_{a|x}\}_a$, such that:
\begin{align}
&M_{a|x} \succeq 0  \qquad &&\forall x \in \X, \forall a \in \A \\
&\sum_{a\in \A} M_{a|x} = \idd  &&\forall x \in \X.
\end{align}

Bob's measurements are defined similarly and denoted by $\{N_{b|y}\}_b$. The probability that Alice and Bob get outputs $a$ and $b$, given inputs $x$ and $y$ is determined by the Born rule $p(a,b|x,y) = \bra{\psi}M_{a|x}\otimes N_{b|y}\ket{\psi}$. In that case, the winning probability is equal to: 
\begin{equation}
    \omega = \sum_{x,y}q(x,y)V(a,b|x,y)\bra{\psi}M_{a|x}\otimes N_{b|y}\ket{\psi}.
\end{equation}
and the maximal quantum score is given by:
\begin{equation}
    \omega_q = \max_{\{M_{a|x}\}_x,\{N_{b|y}\}_y,\ket{\psi}}\omega.
\end{equation}
Note that, more generally, Alice and Bob could have access to a mixed state $\rho$, but since we don't impose any restriction on the dimension of the underlying Hilbert space, any score reachable with a mixed state can also be reached with a pure state by increasing the dimension of the Hilbert space.   

If  players output bits, i.e. $\A=\B=\{0,1\}$, the game is said to be binary, and the quantum strategies can be characterised in a simpler way: their measurements are defined by the Hermitian measurement observables
\begin{align}\label{eq:binary_observables}
    A_x = \sum_a(-1)^aM_{a|x}, &\qquad B_y = \sum_b(-1)^bN_{b|y}.
\end{align}
In the rest of the paper, we use the following notations:
\begin{align}
    \ket{A_x} = A_x\ket{\psi}, &\qquad \ket{A} = \sum_x\ket{A_x}\otimes\ket{x},\\
    \ket{B_y} = B_y\ket{\psi}, &\qquad \ket{B} = \sum_y\ket{B_y}\otimes\ket{y},
\end{align}
which allows us to express the players' correlations through a correlation matrix as:
\begin{equation}
    C = \sum_{x,y}c_{xy}\ketbra{x}{y} \qquad \text{with } c_{xy} = \braket{A_x}{B_y}.
\end{equation}

Nonlocal games are closely related to Bell inequalities~\cite{SMA2008Relation, AH2012Logical}. A Bell inequality can be defined via a linear form on the space of the conditional distributions $p(a,b|x,y) \in \mathbb{R}^{|A||B||X||Y|}$:
\begin{equation}
\sum_{a,b,x,y} \gamma_{a,b,x,y} p(a,b|x,y) \leq \beta_c,
\end{equation}
where $\{ \gamma_{a,b,x,y} \}_{a,b,x,y}$ is a fixed set of complex coefficients; different coefficients identifies different inequalities.

Importantly, the quantum Bell score can be written as $\beta = \bra{\psi}\mathcal{B}\ket{\psi}$ where $\mathcal{B}$ is the Bell operator
\begin{equation}
    \mathcal{B} = \sum_{a,b,x,y}\gamma_{a,b,x,y}M_{a|x}\otimes N_{b|y}.
\end{equation}
When there is a gap between the maximal classical and quantum scores, the quantum strategies that achieve a score higher than the classical value also violate a Bell inequality. For instance, the quantum state and measurements yielding the maximal value for the CHSH game also yield the maximal quantum value, known as Tsirelson's bound, for the CHSH inequality. In the remainder of this paper, we sometimes switch between the nonlocal game and the Bell inequality formulations. The local bound of a Bell inequality will be denoted with $\beta_c$, while the maximal value considering quantum strategy will be denoted as $\beta_q$. For a Bell inequality with quantum bound $\beta_q$ we will often use the shifted Bell operator
\begin{equation}
    \beta_q\idd - \mathcal{B},
\end{equation}
which is by construction positive semi-definite.

\subsection{XOR games}\label{sec:xor}

XOR games are a subclass of nonlocal games with binary outputs, in which the winning conditions for each pair of inputs depend only on the parity of the two output bits, i.e. there is a function $f$ such that $V(a,b,x,y) = 1$ if and only if $f(x,y) = a\oplus b$. Hence, the winning probability takes the form
\begin{equation}
    \omega = \sum_{a,b,x,y}q(x,y) \delta_{a\oplus b = f(x,y)} p(a, b|x,y).
    \footnote{Here $\delta$ denotes the Kronecker function.}
\end{equation}
The winning probability of an XOR game is connected to its bias 
\begin{equation}
    \xi = \sum_{x,y}q(x,y)(-1)^{f(x,y)}c_{x,y},
\end{equation}
through the relation $\omega = (1+\xi)/2$. The term "bias" comes from the fact that in XOR games, the score $1/2$ can be achieved by players that do not use any resources, but randomly output their bits, independently of the inputs. An XOR game can be characterized by its so-called game matrix $\Phi$:
\begin{equation}
    \Phi = \sum_{x,y}q(x,y)(-1)^{f(x,y)}\ket{x}\bra{y}.
\end{equation}
The bias can then be compactly written as $\xi = \tr\left[C^T\Phi\right]$. We can also define the operator associated to an XOR game as 
\begin{equation}\label{eq:XOR_bell_operator}
    \mathcal{B}_g = \sum_{x,y}q(x,y)(-1)^{f(x,y)}A_x\otimes B_y,
\end{equation}
and the bias can then be computed as the expectation value of the corresponding game operator $\xi = \bra{\psi}\mathcal{B}_g\ket{\psi}$. From this game operator, one obtains a valid Bell operator by simply multiplying the r.h.s. of eq.~\eqref{eq:XOR_bell_operator} with $\left|\X\right|\cdot\left|\Y\right|$.

\begin{lem}[\cite{tsirelson}]\label{lemma1}
    Alice's optimal quantum strategy in an XOR game, encoded in the vector $\ket{A_q}$, is fully determined by Bob's optimal strategy $\ket{B_q}$ through a linear transformation:
    \begin{equation}
        \ket{A_q} = \idd\otimes F\ket{B_q}.
    \end{equation}
\end{lem}
\begin{proof}
This lemma was proven by Tsirelson in~\cite{tsirelson}, but we reproduce here the proof of~\cite{escola}. The optimal quantum bias of an XOR game can be obtained as the solution to a semidefinite program (SDP)~\cite{wehner}. For that purpose let us define the Gram matrix $\tilde{Q}$
\begin{equation}
    \tilde{Q} = \left[
\begin{array}{c|c}
R  & C \\ \hline
 C^T & S
\end{array}\right], \textrm{ with } R = \sum_{x,x'}\braket{A_x}{A_{x'}}\ketbra{x}{x'}, \textrm{ and } S = \sum_{y,y'}\braket{B_y}{B_{y'}}\ketbra{y}{y'}.
\end{equation}
Every Gram matrix is positive semidefinite and has $1$-s on the diagonal. Thus, the SDP yielding the optimal quantum bias for an XOR game takes the form
\begin{align}\nonumber
    \xi_q = &\max_{\tilde{Q}}\tr\left[\tilde{Q}\tilde{\Phi}\right],\\ \label{sdpprimal}
    \mathrm{s.t.}\qquad  &\tilde{Q}_{ii} = 1,\quad \mathrm{for}\quad i = 1,\cdots,\x+\y,\\ \nonumber
    &\tilde{Q} \succeq 0,
\end{align}
where $\tilde{\Phi} = \frac{1}{2}\begin{pmatrix}
    0 & \Phi\\
    \Phi^T & 0
\end{pmatrix}$. From this primal form of the SDP, we can obtain the dual formulation by introducing the Lagrangian $\mathcal{L} = \tr[\tilde{Q}\tilde{\Phi}] - \frac{1}{2}\sum_i\lambda_i(\tr \left[\proj{i}\tilde{Q}\right] - 1)$, where $\lambda_i$ are nonnegative Lagrangian multipliers. Let us define the diagonal matrix $\Lambda = \mathrm{diag}(\lambda_1,\cdots,\lambda_{\x+\y})$. The minimal value of the Lagrangian upper bounds the solution to the primal SDP if $\frac{1}{2}\Lambda - \tilde{\Phi}\succeq 0$, so the dual SDP takes the form
\begin{align}\label{dual}
    \min\tr\left[\frac{\Lambda}{2}\right], \qquad \mathrm{s.t.}\quad \frac{\Lambda}{2} -\tilde{\Phi}\succeq 0.
\end{align}
As a consequence of the aforementioned construction, it can be deduced that  any given pair of primal and dual feasible solutions satisfy $\tr[\tilde{Q}\tilde{\Phi}] \leq \xi_q 
\leq \tr[\Lambda]/2$. The optimal values of the primal and of the dual coincide if strong duality holds. A sufficient condition for strong duality to hold, when the primal and dual are finite, is the existence of a strictly feasible solution to the dual problem, which is satisfied here.  Hence, the optimal value can be obtained if the complementary slackness condition is satisfied
\begin{equation}\label{compslack}
    \tr\left[\tilde{Q}\left(\frac{1}{2}\Lambda - \tilde{\Phi}\right)\right] = 0.
\end{equation}

Let us now introduce the notations $\Lambda = \Lambda_A\oplus\Lambda_B $, where $\Lambda_A$ ($\Lambda_B$) is a diagonal $\x\times \x$ $\left(\y\times \y\right)$ matrix, and $|Q\rangle = \ket{A}\oplus \ket{B}$. We can write the quantum bias as $\xi = \bra{Q}\idd\otimes \tilde{\Phi}\ket{Q}$. The slackness condition~\eqref{compslack} gives $\tr\left[\idd\otimes\left(\frac{1}{2}\Lambda - \tilde{\Phi}\right)\proj{Q}\right] = 0$, or equivalently $\idd\otimes(\frac{1}{2}\Lambda -\tilde{\Phi})\ket{Q} = 0$, since $\frac{1}{2}\Lambda - \tilde{\Phi}$ is positive semidefinite. The form of $\Lambda$ and $\tilde{\Phi}$ implies the following conditions for the optimal strategies of Alice and Bob:
\begin{equation}\label{ab}
    \idd\otimes \Lambda_A\ket{A_q} = \idd\otimes\Phi\ket{B_q}, \quad \mathrm{and}\quad \idd\otimes \Lambda_B\ket{B_q} = \idd\otimes\Phi^T\ket{A_q}.
\end{equation}
Taking $F = \Lambda_A^{-1}\Phi$ gives the desired result.
\end{proof}
Another convenient way to write the relations obtained above is
\begin{equation}\label{ax}
    \ket{A_x} = \lambda_{x}^{-1}\sum_{y=1}^{m_B}\Phi_{xy}\ket{B_y}, \qquad \mathrm{for}~x = 1,\cdots,\x.
\end{equation}
The matrix $\tilde{Q}$ can be seen as a moment matrix corresponding to the first level of Navascues-Pironio-Ac\'{\i}n hierarchy (NPA)~\cite{NPA}. The $n$-th level NPA moment matrix is obtained by defining all degree-$n$ monomials $S_i$ of the operators from the set $\{A_1,\cdots,A_{\x},B_1,\cdots,B_{\y}\}$ and taking $\tilde{Q}^{(n)} = \bra{\psi}S_i^{\dagger}S_j\ket{\psi}\ketbra{i}{j}$. Upper bounds on the maximal quantum score of an arbitrary nonlocal game can be obtained as solutions to an SDP analogous to~\eqref{sdpprimal}, but taking $\tilde{Q}^{n}$ instead of $\tilde{Q}$ and the appropriate game matrix $\tilde{\Phi}$. The larger $n$ is, the tighter the upper bound is. In the case of XOR games, the moment matrix of the first level of the NPA hierarchy actually suffices to find the exact optimal value, as it was proven in~\cite{navascues2010glance}. The dual formulation can be seen as an optimization over sum-of-squares (SOS) polynomials, given the duality theory between positive semidefinite moment matrices and SOS polynomials~\cite{laurent2009sums}. In the case of the NPA hierarchy the difference between the solutions of the dual and primal problems can be seen as the expectation value of an SOS polynomial $\sum_{i}S_i^\dagger S_i$, where $S_i$ belongs to the monomials used to create the corresponding moment matrix~\cite{tavakoli2024semidefinite}.  For XOR games, this implies that for every Gram matrix $\tilde{Q}$ obtained by measuring the state $\ket{\psi}$, the following holds:
\begin{equation}\label{sosnum}
    \xi_q - \tr[\tilde{Q}\tilde{\Phi}] = \bra{\psi}\sum_{i}P_i^\dagger P_i\ket{\psi},
\end{equation}
where the $P_i$-s are first degree polynomials over $\{A_1,\cdots,A_{\x},B_1,\cdots,B_{\y}\}$, i.e. $P_i = \sum_{x,y}(\alpha_x^i A_x + \beta_y^i B_y)$. As Eq.~\eqref{sosnum} holds for all quantum realizations, it can be written as
\begin{equation}
    \label{sosop}\xi_q\idd - \mathcal{B}_g = \sum_{i}P_i^\dagger P_i,
\end{equation}
where $\xi_q\idd - \mathcal{B}_g$ is usually called the shifted game operator. The following theorem stipulates that for XOR games, the shifted game operator can be written as a sum of squares, with each term containing a single Alice operator, plus a positive polynomial depending only of Bob's operators.

\begin{thm}\label{sosThm}
    Let $\mathcal{B}_g$ be the game operator of an XOR game with optimal quantum bias $\xi_q$. Then the following holds:
\begin{equation}\label{sosXOR}
    \xi_q\idd - \mathcal{B}_g = \sum_x\frac{\lambda_{x}}{2}\left(A_x - \sum_{y}F_{xy}B_y\right)^2 + P\left(\{B_y\}_y\right),
\end{equation}
where $F$ is the matrix of Lemma~\ref{lemma1} and $P(\{B_y\}_y)$ is a positive polynomial over Bob's measurement operators.
\end{thm}
To prove this theorem, we first use the following lemma proven by Ostrev.

\begin{lem}\cite{ostrev2016structure}\label{ostrevlemma}
   \begin{enumerate}
       \item 
    Let $\lambda_1,\cdots,\lambda_{\x+\y}$ be an optimal solution to the dual semidefinite program~\eqref{dual}. Then there exist vectors $\{\ket{u_i} = \sum_{j=1}^{\x}u_{ij}\ket{j}\}_{i=1}^r$ 
 and $\{\ket{v_i} = \sum_{j=1}^{\y} v_{ij}\ket{j}\}_{i=1}^r$ such that
    \begin{align}\label{ostrev}
        \sum_{i=1}^r\proj{u_i} = \frac{1}{2}\Lambda_A, \qquad \sum_{i=1}^r\proj{v_i}= \frac{1}{2}\Lambda_B, \qquad \sum_{i=1}^r\ketbra{u_i}{v_i} = \frac{1}{2}\Phi.
    \end{align}
    \item Let $\ket{A}$ and $\ket{B}$ be a quantum strategy for an XOR game. Let  $\{\ket{u_i} = \sum_{j=1}^{\x}u_{ij}\ket{j}\}_{i=1}^r$ 
 and $\{\ket{v_i} = \sum_{j=1}^{\y} v_{ij}\ket{j}\}_{i=1}^r$  satisfy~\eqref{ostrev}. Then the following identity holds:
   \begin{equation}\label{ostrev1}
       \sum_{i=1}^r\left\|\sum_{j=1}^{\x}u_{ij}\ket{A_j} - \sum_{j=1}^{\y}v_{ij}\ket{B_j}\right\|^2 = \frac{1}{2}{\tr}[\Lambda] - \sum_{ij}\bra{A_i}\Phi_{ij}\ket{B_j}.
   \end{equation}
\end{enumerate}
\end{lem}
\begin{proof}
The interested reader can find the proof of the first part in~\cite{ostrev2016structure} (Lemma 4 therein), we reproduce here the proof of the second part (Lemma 5 in~\cite{ostrev2016structure}). The l.h.s. of~\eqref{ostrev1} reads
\begin{align*}
 &\sum_{i=1}^r\left\|\sum_{j=1}^{\x}u_{ij}\ket{A_j} - \sum_{j=1}^{\y}v_{ij}\ket{B_j}\right\|^2 =\\
 &= \sbra{A}\idd \otimes\left(\sum_{i=1}^r\proj{u_i}\right)\ket{A} + \bra{B}\idd\otimes\left(\sum_{i=1}^r\proj{v_i}\right)\ket{B} - 2\bra{A}\idd\otimes\left(\sum_{i=1}^r\ketbra{u_i}{v_i}\right)\ket{B}=\\
 &= \frac{1}{2}\sbra{A}\idd \otimes\Lambda_A\ket{A} + \frac{1}{2}\bra{B}\idd\otimes\Lambda_B\ket{B} - \bra{A}\left(\idd\otimes\Phi\right)\ket{B}=\\
 &= \frac{1}{2}\tr[\Lambda_A] + \frac{1}{2}\tr[\Lambda_B] - \bra{A}\idd\otimes\Phi\ket{B},
\end{align*}
which is exactly the second statement of the lemma, given that $\Lambda = \Lambda_A + \Lambda_B$.
\end{proof}

We now define a robust version of Ostrev's lemma.

\begin{lem}\label{ostrevlemmarob}
\begin{enumerate}
    \item Let $\lambda_1,\cdots,\lambda_{\x+\y}$ be an optimal solution to the dual semidefinite program~\eqref{dual}. Then there exist vectors $\{\ket{u_i} = \sum_{j=1}^{\x}u_{ij}\ket{j}\}_{i=1}^r$ 
    and $\{\ket{v_i} = \sum_{j=1}^{\y} v_{ij}\ket{j}\}_{i=1}^r$ such that
    \begin{align}\label{ostrev2}
        \sum_{i=1}^r\proj{u_i} = \frac{1}{2}\Lambda_A, \qquad \sum_{i=1}^r\proj{v_i} \preceq \frac{1}{2}\Lambda_B, \qquad \sum_{i=1}^r\ketbra{u_i}{v_i} = \frac{1}{2}\Phi
    \end{align}
    \item Let $\ket{A}$ and $\ket{B}$ be a quantum strategy for an XOR game. Let  $\{\ket{u_i} = \sum_{j=1}^{\x}u_{ij}\ket{j}\}_{i=1}^r$ 
    and $\{\ket{v_i} = \sum_{j=1}^{\y} v_{ij}\ket{j}\}_{i=1}^r$  satisfy~\eqref{ostrev2}. Then the following inequality holds:
   \begin{equation}\label{ostrev3}
       \sum_{i=1}^r\left\|\sum_{j=1}^{\x}u_{ij}\ket{A_j} - \sum_{j=1}^{\y}v_{ij}\ket{B_j}\right\|^2 \leq \frac{1}{2}{\tr}[\Lambda] - \sum_{ij}\bra{A_i}\Phi_{ij}\ket{B_j}.
       \end{equation}
\end{enumerate}
\end{lem}

\begin{proof}
For the proof of the first part, we give an explicit construction of the vectors $\ket{u_i}$ and $\ket{v_i}$.
We choose  $u_{ij}= \sqrt{\frac{\lambda_i}{2}}\delta_{ij}$, where $\delta_{ij}$ is the Kronecker delta, and we take $\ket{v_{i}} = \frac{1}{\lambda_i}\sum_j\Phi_{ij}\ket{u_j}$.
The first and third relations in~\eqref{ostrev2} are satisfied, as $\sum_i\proj{u_i} = \frac{1}{2}\mathrm{diag}(\lambda_1,\cdots,\lambda_{\x})$, and $\sum_i\ketbra{u_i}{v_i} = \frac{1}{2}\Phi$.
Concerning the second relation, we get $\sum_{i}\proj{v_i} = \frac{1}{2}\Phi^T\Lambda_A^{-1}\Phi$.
To analyze it, we use the complementary slackness condition, which implies that $\Lambda - \tilde{\Phi} \succeq 0$.
By Schur's complement lemma , this matrix is positive if and only if $\Lambda_B - \Phi^T \Lambda_A^{-1} \Phi \succeq 0$.
This completes the proof of the first part of the lemma.
The second part follows analogously to the second part of Lemma~\ref{ostrevlemma}.
\end{proof}

We can now prove Theorem~\ref{sosThm}. 

\begin{proof}

Taking the $\ket{u_i}$-s and $\ket{v_i}$-s used to prove the first part of Lemma~\ref{ostrevlemmarob} in~\eqref{ostrev3}, we get
\begin{align}
   \xi_q - \bra{\psi}\mathcal{B}_g\ket{\psi} - \sum_{x=1}^{m_A}\frac{\lambda_x}{2}\left\|A_x\ket{\psi} - \lambda_x^{-1}\sum_y\Phi_{x,y}B_y\ket{\psi}\right\|^2 \geq  0.
\end{align}
When opening the sum of squares, we get
\begin{equation}
   \xi_q - \sum_{x=1}^{m_A}\frac{1}{\lambda_x}\bra{\psi}\sum_{y,y'}B_y\Phi_{y,x}\Phi_{x,y'}B_{y'}\ket{\psi} \geq 0,
\end{equation}
which holds for every valid quantum state $\ket{\psi}$, implying 
\begin{equation}\label{PBpositive}
    2P\left(\{B_y\}_y\right) \equiv \xi_q\idd - \sum_{x=1}^{m_A}\frac{1}{\lambda_x}\sum_{y,y'}B_y\Phi_{y,x}\Phi_{x,y'}B_{y'} \succeq 0.
\end{equation}
The simple expansion of all the squares shows that
\begin{equation}
    \xi_q\idd - \mathcal{B}_g = \frac{\lambda_x}{2}\left(A_x - \lambda_x^{-1}\sum_y\Phi_{x,y}B_y\right)^2 + P\left(\{B_y\}_y\right),
\end{equation}
which together with~\eqref{PBpositive} proves the theorem. 
\end{proof}

\subsection{Self-testing}

Self-testing is a powerful technique that allows one to identify quantum states and measurements leading to the optimal winning probability in a nonlocal game. To be more precise, the score $\omega_q$ is said to self-test the state and measurements $\ket{\psi'},\{M'_{a|x}\},\{N'_{b|y}\}$ if for all states and measurements $\physstatedm,\{M_{a\vert x}\}, \{N_{b\vert y}\}$ allowed to reach $\omega_q$ there exists 
\begin{itemize}
\item[(i)] local Hilbert spaces $\mathcal{H}_\rA$, $\mathcal{H}_\rB$ such that $\physstatedm\in\mathcal{L}[\mathcal{H}_\rA\tp\mathcal{H}_\rB]$, $M_{a\vert x}\in\mathcal{L}[\mathcal{H}_\rA]$, $N_{b\vert y}\in\mathcal{L}[\mathcal{H}_\rB]$ 
\item[(ii)] a local isometry $\iso=\iso_\rA\tp\iso_\rB$ 
\end{itemize}
such that for any purification $\ket{\psi}^{\rA\rB\rP}$ of $\varrho^{\rA\rB}$ it holds
\begin{align}\nonumber
    \iso\otimes \mathbb{I}_{\rP}\left[M_{a\vert x}\tp N_{b\vert y}\tp\mathbb{I}_\rP\ket{\psi}^{\rA\rB\rP}\right]=\left(M'_{a\vert x}\tp N'_{b\vert y}\ket{\psi'}^{\rA'\rB'}\right)\otimes\junk^{\rA\rB\rP}
\end{align}
for all $a,x,b,y$ and for some state $\junk^{\rA\rB\rP}$.
In this paper, we are especially interested in self-testing of measurements. A compact way of writing the self-testing statement for Bob's measurements from the description given above is that there exists an isometry $\Phi_B$ that maps each $M_{b|y}$ to $M'_{b|y}\otimes\idd$. 

\subsection{Quantum homomorphic encryption}

Homomorphic encryption is a cryptographic technique that enables to execute computations directly on encrypted data, without prior decryption. The results of these computations remain in an encrypted form and, upon decryption, yield outputs identical to those obtained through operations on the unencrypted data. The term "homomorphic" draws from algebraic homomorphism, wherein encryption and decryption functions are likened to homomorphisms between plaintext and ciphertext spaces. A cryptosystem supporting arbitrary computation on ciphertexts is termed fully homomorphic encryption (FHE), representing the most robust form of homomorphic encryption. Originally conceptualized as a privacy homomorphism by Rivest, Adleman, and Dertouzous~\cite{rivest1978data} shortly after the invention of the RSA cryptosystem~\cite{rivest1978method}, the first plausible construction for FHE using lattice-based cryptography was presented by Gentry~\cite{STOC:Gentry09}. 
Leveraging the hardness of the (Ring) Learning With Errors (RLWE) problem, more efficient schemes for fully homomorphic encryption have been devised~\cite{ITCS:BraGenVai12,brakerski2014efficient}. The possibility of quantum homomorphic encryption (QHE), allowing for quantum computations on encrypted data, was introduced by  Mahadev~\cite{FOCS:Mahadev18b}, with Brakerski~\cite{C:Brakerski18} subsequently enhancing it to achieve improved noise tolerance. 

Before reminding the formalism of QHE, we first recall the definition of quantum polynomial time algorithms.
Throughout the paper, $\kappa$ denotes the security parameter.
\begin{defi}[Quantum polynomial time algorithm]
    A quantum algorithm is quantum polynomial time (QPT) it can be implemented by a family of quantum circuits with size polynomial in the security parameter $\kappa$.
\end{defi}
We now reproduce the definition of QHE as it appears in~\cite{STOC:KLVY23}.

\begin{defi}[Quantum Homomorphic Encryption (QHE)]\label{def:QHE-aux}
A quantum homomorphic encryption scheme $\mathsf{QHE}=(\gen, \enc, \eval, \dec)$ for a class of quantum circuits $C$ consists of  the following four quantum algorithms which run in quantum polynomial time in terms of the security parameter:
\begin{itemize}
    \item ${\gen}$ takes as input the security parameter $1^\kappa$ and outputs a (classical) secret key ${\sk}$ of size $\mathrm{poly}(\kappa)$ bits;
    \item ${\enc}$ takes as input a secret key ${\sk}$ and a classical input $x$, and outputs a ciphertext ${\ct}$;
    \item ${\eval}$ takes as input a tuple $(C,\ket{\Psi},{\ct}_{\mathrm{in}})$, where $C:\mathcal{H} \times (\mathbb{C}^2)^{\otimes n}\rightarrow (\mathbb{C}^2)^{\otimes m}$ is a quantum circuit, $\ket{\Psi}\in\mathcal{H}$ is a quantum state, and ${\ct}_{\mathrm{in}}$ is a ciphertext corresponding to an $n$-bit plaintext. 
    ${\eval}$ computes a quantum circuit
    ${\eval}_C(\ket{\Psi} \otimes \ket{0}^{\otimes \mathrm{poly}(\kappa, n)},{\ct}_{\mathrm{in}})$ which outputs a ciphertext ${\ct}_{\mathrm{out}}$. If $C$ has classical output, we require that ${\eval}_C$ also has classical output.
    \item ${\dec}$ takes as input a secret key $\sk$ and ciphertext ${\ct}$, and outputs a state $\ket{\phi}$. Additionally, if ${\ct}$ is a classical ciphertext, the decryption algorithm outputs a classical string $y$.
\end{itemize}
\end{defi}
As in~\cite{STOC:KLVY23} the following property is required from $\mathsf{QHE}$, in order for it to behave ``nicely" with entanglement:
\begin{defi}
[Correctness with auxiliary input] For every security parameter $\kappa\in\mathbb{N}$, any quantum circuit $C:\mathcal{H}_{A} \times (\mathbb{C}^2)^{\otimes n} \to \{0,1\}^*$ (with classical output), any quantum state $\ket{\Psi}_{AB} \in\mathcal{H}_{A} \otimes \mathcal{H}_{B}$,  any message $x\in \{0,1\}^n$, any secret key $\sk \gets \gen(1^\kappa)$ and any ciphertext $\ct \gets \enc(\sk,x)$, the following states have negligible trace distance:
    \begin{description}
        \item[{Game $1$}] Start with $(x, \ket{\Psi}_{AB})$. Evaluate $ C$ on $x$ and register $A$, obtaining classical string $y$. Output $y$ and the contents of register $B$.
        \item[{Game $2$}] Start with $\ct \gets \enc(\sk, x)$ and $\ket{\Psi}_{AB}$. Compute  $\ct' \gets \eval_C(\ket{\Psi}_{AB} \otimes \ket{0}^{\mathrm{poly}(\kappa, n)},\ct)$ on register $A$. Compute $y'= \dec(\sk,\ct')$. Output $y'$ and the contents of register $B$.
    \end{description}
\end{defi}    

In simple terms, ``correctness with auxiliary input" stipulates that when QHE evaluation is employed on a register A within a joint (entangled) state in $\mathcal{H}_{A}\otimes\mathcal{H}_{B}$, the entanglement between the QHE-evaluated output and B must be maintained.

Finally, the following definition characterizes another property expected from QHE, and it is in cryptography well-known indistinguishability under chosen plaintext attack (IND-CPA).

\begin{defi}
      An QHE scheme is IND-CPA secure if for every quantum-polynomial time (QPT)
      
  adversary \(\mathcal{A} = (\mathcal{A}_1, \mathcal{A}_2)\), there exists a negligible
  function \(\mathsf{negl}\) such that the following holds for all $\kappa \in \mathbb{N}$:
\begin{align*}
\Pr\left[
\Tilde{b} =b 
\ \middle\vert
\begin{array}{ll}
\sk \gets \gen(1^{\kappa})  \\
(x_0, x_1) \gets \mathcal{A}_1^{\enc(\sk,\cdot)}(1^\kappa)  \\
b\xleftarrow{\$} \{0,1\} \\ 
\ct^* \gets \enc(\sk, x_b)\\  
\tilde{b} \gets \mathcal{A}_2^{\enc(\sk,\cdot)}(\ct^*)
\end{array}
\right] \leq \frac{1}{2} + \mathsf{negl}(\kappa).
\end{align*}
\end{defi}

\subsection{Compiled nonlocal games}

In~\cite{STOC:KLVY23}, Kalai \textit{et al.} introduced a compilation technique that can be used to construct single-prover proofs of quantumness. Their procedure transforms any $k$-player nonlocal game into a single-prover interactive game, employing post-quantum cryptography to emulate spatial separation among the parties. The proposed protocol maintains classical soundness, ensuring that no classical polynomially-bounded prover can surpass the maximal classical score of the original game. Additionally, leveraging quantum homomorphic encryption (QHE), the authors devise an explicit and efficient quantum strategy that achieves the quantum bound of the original nonlocal game. This enables the translation of proofs of quantumness into the single-prover interactive proof framework.

In this context, we refer to a single-prover interactive game that is generated through the KLVY compilation of a nonlocal game as a compiled nonlocal game. 
We recall the definition of compiled nonlocal games introduced in~\cite{STOC:KLVY23}. Using the interactive proof terminology, the entity called ``referee'' in the nonlocal game will be referred to as the ``verifier''.

\begin{defi}[Compiled nonlocal game]
\label{def:cspgame}
    In a compiled nonlocal game, a verifier, equipped with access to a Quantum Homomorphic Encryption (QHE) scheme as defined in Def.~\ref{def:QHE-aux}, engages with a prover. According to a probability distribution $q(x, y)$, the verifier samples $x$ and $y$. In the first round, the verifier transmits $\xe = \enc(x)$ to the prover, who responds with an encrypted output $\aen = \enc(a)$. In the second round, the verifier sends the input $y$ to the prover in the clear, and the prover replies with the answer $b$. The verifier assesses the outcome using the game predicate $V(\dec(\aen), b|\dec(\xe), y) \in \{0,1\}$ to determine whether the prover has passed or failed in the game.
\end{defi}

The following theorem (Theorem 1.1 in~\cite{STOC:KLVY23}) relates the classical and quantum score of a nonlocal game to the scores of the corresponding compiled game.
\begin{thm}[\cite{STOC:KLVY23}]
Given any $2$-player nonlocal game with quantum bound $\xi_q$ and classical bound $\xi_c$ and  any QHE scheme (with security parameter $\kappa$) that satisfies correctness with auxiliary inputs and  IND-CPA security against quantum distinguishers
then, there is a $4$-round single prover interactive game with completeness $\xi_q$ realized by a quantum polynomial-time algorithm
and soundness $\xi_c + \text{negl}(\kappa)$ against any classical polynomial-time algorithm.
\end{thm}

The classical soundness statement in this theorem guarantees that the maximal winning probability that a classical polynomial-time prover can achieve in a compiled nonlocal game is nearly identical to the optimal classical winning probability in the corresponding nonlocal game, with a negligible deviation dependent on the security parameter.  The main insight in~\cite{STOC:KLVY23} is that the success of the classical prover is primarily hindered by the sequential nature of the game. This sequential structure forces the prover to commit to an answer $\aen$ before receiving the input $y$. The combination of this sequential setup with secure QHE effectively replicates the locality requirement of a nonlocal game that is ensured by the spatial separation of the players. The quantum completeness statement ensures that a QPT prover can win the compiled nonlocal game with a probability that is at least as large as the optimal quantum winning probability in the corresponding nonlocal game.

\subsubsection{Modelling the quantum prover}
Let us now model the behavior of the single quantum prover, in the same way as it was done in~\cite{FOCS:NatZha23}. In a compiled game, denoted as per Def.~\ref{def:cspgame}, the prover, initially in state $\ket{\psi}$, undergoes a process involving encrypted questions and answers. Specifically, in the first round, the prover receives an encrypted question $\xe$, performs a POVM measurement, and computes an encrypted answer $\aen$. Using Naimark dilation theorem, the prover's POVM measurement is simulated by a projective measurement, denoted here with $M_{\aen|\xe}$. The prover's action could potentially involve a unitary operation $U_{\xe,\aen}$ following the measurement, crucial in the sequential setting. The projectors and unitaries can be unified into a set of potentially non-Hermitian operators $\M_{\aen|\xe}$, satisfying $\M_{\aen|\xe}^\dagger\M_{\aen|\xe} = M_{\aen|\xe}$ and hence $\sum_{\aen}\M_{\aen|\xe}^\dagger\M_{\aen|\xe} = \idd$. The prover's state after the first round of the game corresponds to the post-measurement state
\begin{equation}\label{eq:postMstate}
    \ket{\psi_{\aen|\xe}} = \M_{\aen|\xe}\ket{\psi},
\end{equation}
and the probability to get output $\aen$ for input $\xe$ is 
\begin{equation}\label{eq:probenc}
    p(\aen|\xe) = \bra{\psi}\M_{\aen|\xe}^\dagger\M_{\aen|\xe}\ket{\psi} = \left\|\ket{\psi_{\aen|\xe}}\right\|^2.
\end{equation}
In the second round, the prover's behavior is characterized by a set of projective operators $\{\{N_{b|y}\}_b\}_{y}$. If the second-round answers are bits, the measurements can be characterized by specifying a Hermitian observable $B_y = \sum_{b}(-1)^bN_{b|y}$. Similarly, if Alice's outputs in the corresponding nonlocal game are bits, we can define a ``decrypted''  observable
\begin{align}\label{eq:encObs}
    \rA_x &= \mathop{\mathds{E}}_{\xe:\dec(\xe) = x}\sum_{\aen}(-1)^{\dec(\aen)}\M_{\aen|\xe}^\dagger\M_{\aen|\xe}\\
    &= \mathop{\mathds{E}}_{\xe:\dec(\xe) = x}A_\xe,
\end{align}
where $A_\xe$ are binary observables, while $\rA_x$ in general is not. If both $a$ and $b$ are bits, we can define the correlators allowing to characterize the winning probability of a quantum prover in a computational single-prover game:
\begin{equation}\label{eq:cryptocorr}
     \langle\rA_x,B_y\rangle = \mathop{\mathds{E}}_{\xe:\dec(\xe) = x}\sum_{\aen}(-1)^{\dec(\aen)}\bra{\psi_{\aen|\xe}}B_y\ket{\psi_{\aen|\xe}}
\end{equation}
The correlators have the same operational meaning as in nonlocal games: when the verifier samples a question pair $(x,y)$ in the compiled game and receives (decrypted) answers $(a,b)$,  $\langle\rA_x,B_y\rangle$ is precisely the expected value of $(-1)^{a+b}$.

The marginals of the second-round observables in principle depends on the encrypted $x$ of the first round and have form:
\begin{equation}
    \langle B_y\rangle_{x} = \Exp_{\xe:\dec(\xe) = x}\sum_{\aen}\bra{\psi_{\aen|\xe}}B_y\ket{\psi_{\aen|\xe}}.
\end{equation}

\subsection{Technical tools for estimating quantum bounds of compiled nonlocal games}

\subsubsection{Block encodings}

We now examine certain outcomes related to the block encoding of quantum processes. Block encoding is a method for the efficient implementation of a quantum operation. Our motivation for exploring this process stems from the contextual constraint imposed in the subsequent sections of this paper, where we address computational limitations among participants. Consequently, we aim to identify operations that can be efficiently executed by leveraging available quantum resources. 

\begin{defi}[Block encoding] \label{def:bloc}
    Given a matrix $A \in \mathds{C}^{c\times c}$, we say that $U \in \mathds{C}^{d\times d}$ is a $Q$-block encoding of $A$ if
    \begin{itemize}
        \item $U$ is a unitary matrix whose quantum circuit can be implementable with $O(Q)$ gates,
        \item $U$ has the following form $U = \begin{pmatrix}
            \tau A & \cdot \\
            \cdot & \cdot
        \end{pmatrix}$, where we call $\tau$ the scale factor of the block encoding. 
    \end{itemize}
We say $U$ is QPT-implementable if $Q$ is polynomially bounded in the size of the input. 
\end{defi}
If $U$ is a block encoding of $A$, then $A$ can be implemented by performing the following operation 
\begin{equation*}
    A = \frac{1}{\tau }\left(\bra{0}\otimes\idd\right)U(\ket{0}\otimes\idd)
\end{equation*}
where $\idd$ is the identity matrix of the size of $A$. 

Through the linear combination and multiplication of matrices possessing a block encoding, our anticipation is the persistence of this property in the resulting matrix. Although the explicit construction in complete generality is not immediately apparent, \cite[Lemmas 52 and 53]{STOC:GSLW19} provides a technical framework for the block encoding of linear combinations and products of matrices with block encoding. Here, we present a streamlined version of their results, tailored to our specific requirements.

First, we present a lemma about the block encoding of a linear combination of matrices that have block encoding, whose proof can be found in~\cite[Lemma 52]{STOC:GSLW19}.

\begin{lem}[{Linear combination of block encoded matrices}]
    Let $A= \sum_{j=1}^m y_j A_j$, where $y \in \mathds{C}^m$ is a complex bounded vector $\|y\|_1 \leq \beta$ and $A_j$ are matrices for which we know a $Q_j$-block encoding $U_j$.
    Then we can implement a $\left(m + \sum_{j=1}^m Q_j\right)$-block encoding of $A$.
\end{lem}

\begin{corollary}
\label{cor:add-block-encoding}
    Let $A= \sum_{j=1}^m y_j A_j$, where $y \in \mathds{C}^m$ is a complex bounded vector $\|y\|_1 \leq \beta$. 
    If $m = \mathrm{poly}(\kappa)$ and each $A_j$ is an operator with QPT-implementable block encodings with scale factor $O(1)$ for all $j$, $A$ also has a QPT-implementable block encoding with scale factor $O(1)$.
\end{corollary}

The subsequent lemma guarantees the presence of a block encoding for the product of matrices, each possessing its own block encoding, whose proof can be found in~\cite[Lemma 53]{STOC:GSLW19}.
    
\begin{lem}[{Product of block encoded matrices}]
    Let $U$ and $V$ be the $Q_U$ and $Q_V$-block encodings of $A$ and $B$ respectively.
    Then we can implement a $\left(Q_U + Q_V\right)$-block encoding of $A B$.
\end{lem}

\begin{corollary}
\label{cor:product-block-encoding}
    Let $U$ and $V$ be the $Q_U$ and $Q_V$-block encodings of $A$ and $B$ respectively.
    If $U$ and $V$ are QPT-implementable, each with scale factor $O(1)$, $AB$ also has a QPT-implementable block encoding with scale factor $O(1)$.
\end{corollary}
For operators with a QPT-implementable block-encoding, the following technical lemma from~\cite{FOCS:NatZha23} applies.
\begin{lem}[{\cite[Lemma 14]{FOCS:NatZha23}}]
\label{lem:appx-block-encoding}
    Let $\mathscr{B}$ be an operator with a QPT-implementable block encoding with $O(1)$ scale factor and $\| \mathscr{B}\| \leq O(1)$. 
    Then there exists a QPT-measurable POVM $\{M_\beta\}_\beta$ such that for any state $\rho$, the following holds:
    \begin{equation}
        \left| \sum_\beta \beta \cdot \tr\left[M_\beta \rho\right] - \tr\left[\mathscr{B}\rho\right]\right| \leq \varepsilon,
    \end{equation}
    for any $\varepsilon = 1 / \mathrm{poly}(\kappa)$.
\end{lem}
These results, together with the definition of the IND-CPA security, allow us to state the following:
\begin{lem}[{adapted from \cite[Lemmas 15-17]{FOCS:NatZha23}}]
\label{lem:qhe-bound}
    Let $\mathscr{B}$ be an operator with a QPT-implementable block encoding with $O(1)$ scale factor and $\| \mathscr{B}\| \leq O(1)$.  
    Let $\mathrm{QHE} = (\gen, \enc, \eval, \dec)$ be a secure quantum homomorphic encryption scheme (see Definition~\ref{def:QHE-aux}), let $D_0, D_1$ be any two QPT sampleable distributions over plaintext questions $x$, and let $\ket{\psi}$ be any efficiently preparable state of the prover. Then, for any security parameter $\kappa \in \mathbb{N}$, there exists a negligible function $\delta_{\mathrm{QHE}}(\kappa)$ such that
    \begin{align}
    \left|
    \mathop{\mathds{E}}_{x\leftarrow D_0} \mathop{\mathds{E}}_{\xe=\enc(x)} \sum_\aen \bra{\psi} \M_{\aen|\xe}^{\dagger} \mathscr{B} \M_{\aen|\xe} \ket{\psi} - 
    \mathop{\mathds{E}}_{x\leftarrow D_1} \mathop{\mathds{E}}_{\xe=\enc(x)} \sum_\aen \bra{\psi} \M_{\aen|\xe}^\dagger \mathscr{B} \M_{\aen|\xe} \ket{\psi}
    \right|
    \leq \delta_{\mathrm{QHE}}(\kappa).
    \end{align}
\end{lem}
The proof of this Lemma encompasses Lemmas 15, 16 and 17 in~\cite{FOCS:NatZha23}. Importantly, $\mathscr{B}$ does not have to be a binary observable, it just has to be efficiently implementable by quantum circuits.
This lemma establishes a link between QPT-implementable block encodings and IND-CPA security. 
Looking ahead, using the fact that XOR games and SATWAP Bell inequalities have a nice form of SOS decomposition into linear combination and product of block encodings as presented in previous sections, this lemma will allow us to relate the quantum bound of the compiled nonlocal games to those of the original nonlocal games, at the expense of a negligible security loss.

\subsubsection{Crypto-correlation matrix and pseudo-expectation map}\label{cryptoNPA}

In this section, we revisit and further expound upon the concepts delineated in Section 4.4 of \cite{FOCS:NatZha23}. In that section, the authors introduced a compelling argument to establish the optimal quantum winning probability for the compiled CHSH game through the introduction of a cryptographic SOS decomposition. Our objective is to broaden their findings to encompass a more extensive array of games, namely XOR games and $d$-outcome CHSH games.

Consider an XOR game with an optimal quantum bias $\xi_q$. Its game operator~\eqref{eq:XOR_bell_operator} is 
\begin{equation}
    \mathcal{B}_g 
    = \sum_{x,y} \gamma^{AB}_{x,y} A_x B_y \label{eq:Bell_operator},
\end{equation}
and the corresponding shifted game operator has an SOS decomposition
\begin{align}
    \xi_q \mathds{1} - \mathcal{B}_g
     = \sum_i b_i P_i^\dagger P_i
    +\sum_i d_i [A_{x_i}, B_{y_i}] + \sum_i e_i (\mathds{1} - A_{x_i}^2) + \sum_i f_i (\mathds{1} - B_{y_i}^2),\label{eq:SOS_Bell}
\end{align}
where $b_j \in \mathds{R}^+$ are real positive coefficients for the terms of the SOS decomposition, $d_i$-s could be complex numbers referring to the constraint that Alice's and Bob's operators commute, $e_i$-s and $f_i$-s multiply terms that vanish if the measurements of Alice and Bob are projective. All terms on the r.h.s. of~\eqref{eq:SOS_Bell} except the first one are general constraints that are usually implicitly assumed. Here, the aim is to develop an analogous procedure for bounding the optimal quantum bias of compiled games, so one has to be careful as some of the constraints satisfied in the case of nonlocal games might not be satisfied in the case of compiled games. For example, operations in two rounds of the compiled game do not necessarily commute. 

Analogously to matrix $\tilde{Q}$ from Section~\ref{sec:xor}, expounding on the ideas from~\cite{FOCS:NatZha23}, we define a $(|\X|+|\Y|)\times (|\X|+|\Y|)$ matrix $\tilde{\mathcal{Q}}$ as follows
\begin{equation}
    \tilde{\mathcal{Q}}=\left[
\begin{array}{c|c}
\idd_{|\X|}  & {C} \\ \hline
 {C}^T & {S}
\end{array}\right],
\end{equation}
where $\idd_{|\X|}$ is the $|\X|\times |\X|$ identity matrix, where 
\begin{align}
    {C} &= \sum_{x,y}\langle \rA_x,B_y\rangle \ketbra{x}{y},\\
    {S} &= \sum_{y,y'}\mathop{\mathds{E}}_{x \in \X} \mathop{\mathds{E}}_{\xe: \dec(\xe)=x} \sum_\aen  \bra{\psi_{\aen|\xe}}  B_y B_{y'}  \ket{\psi_{\aen|\xe}}\ketbra{y}{y'},\label{eq:Gamma_BB}
\end{align}
and where $\langle \rA_x,B_y\rangle$ correspond to Eq.~\eqref{eq:cryptocorr}. The matrix $\tilde{\mathcal{Q}}$, unlike $\tilde{Q}$, is not necessarily positive semidefinite, as there is no real consistency in assigning values to its entries. Similarly to $\tilde{Q}$, $\tilde{\mathcal{Q}}$ has ones on the diagonal, in the first block by construction and in the second because
\begin{align*}
    {S}_{y,y} &= \mathop{\mathds{E}}_{x \in \X} \mathop{\mathds{E}}_{\xe: \dec(\xe)=x} \sum_\aen  \bra{\psi_{\aen|\xe}}  \idd  \ket{\psi_{\aen|\xe}}\\
    &= \mathop{\mathds{E}}_{x \in \X} \mathop{\mathds{E}}_{\xe: \dec(\xe)=x} \sum_\aen  \bra{\psi}\M_{\aen|\xe}^\dagger\M_{\aen|\xe}\ket{\psi}\\
    &= \mathop{\mathds{E}}_{x \in \X} \mathop{\mathds{E}}_{\xe: \dec(\xe)=x}   \braket{\psi}{\psi}\\
    &= 1.
\end{align*}

Then, as in \cite{FOCS:NatZha23}, we define a linear operator $\Tilde{\mathds{E}}$ that maps every homogeneous degree-2 polynomial in the variables $A_x,B_y$  to linear combinations of elements of the matrix $\tilde{\mathcal{Q}}$ in the following way: 
\begin{align}\label{pseudoExpAB}
    \Tilde{\mathds{E}}[A_xB_y] = {C}_{x,y},&\qquad  \Tilde{\mathds{E}}[B_y\rA_x] = {C}^{T}_{y,x},\\ \label{pseudoExpAABB}
    \Tilde{\mathds{E}}[A_xA_{x'}] = \delta_{x,x'},&\qquad  \Tilde{\mathds{E}}[B_yB_{y'}] = S_{y,y'}.
\end{align}
For $y=y'$ or $x=x'$ in Eq.~\eqref{pseudoExpAABB}, we get a consistent mapping of identity
\begin{equation}\label{xxyy}
    \Tilde{\mathds{E}}[A_xA_{x}] =  \Tilde{\mathds{E}}[\idd] = 1,\qquad  \Tilde{\mathds{E}}[B_yB_y] = \Tilde{\mathds{E}}[\idd] = S_{y,y} = 1.
\end{equation}
As in~\cite{FOCS:NatZha23}, we call the map $\Tilde{\mathds{E}}$ a pseudo-expectation. Such defined pseudo-expectation maps the game operator $\mathcal{B}_g$ introduced in Eq.~\eqref{eq:Bell_operator}) to a bias in the compiled nonlocal game:
\begin{equation}
    \Tilde{\mathds{E}}[\mathcal{B}_g] = \sum_{x,y}\gamma_{x,y} \Tilde{\mathds{E}}[A_xB_y] = \sum_{x,y}\gamma_{x,y}\langle\rA_x,B_y\rangle = \bar{\xi}.
\end{equation}
The optimal quantum bias can be upper bounded using the SOS decomposition of the shifted game operator(eq.~\eqref{eq:SOS_Bell}):
\begin{align}
    \Tilde{\mathds{E}}\left[\xi_q \mathds{1} - \mathcal{B}_g\right] = \Tilde{\mathds{E}}\left[\sum_i b_i P_i^\dagger P_i
    +\sum_i d_i [A_{x_i}, B_{y_i}] + \sum_i e_i (\mathds{1} - A_{x_i}^2) + \sum_i f_i (\mathds{1} - B_{y_i}^2)\right]
\end{align}
Given~\eqref{pseudoExpAB}, the terms multiplied by $d_i$ vanish because, under the pseudo-expectation map, the operators $A_x$ and $B_y$ commute. The terms multiplied by $e_i$ and $f_i$ also vanish because of Eq.~\eqref{xxyy}; hence
\begin{align}
    \Tilde{\mathds{E}}\left[\xi_q \mathds{1} - \mathcal{B}_g\right] = \sum_i b_i \Tilde{\mathds{E}}\left[P_i^\dagger P_i
    \right].
\end{align}
Thus, if $\Tilde{\mathds{E}}\left[P_i^\dagger P_i \right]$ is non-negative, the bias in a compiled game cannot be larger than the optimal quantum bias of the corresponding nonlocal game.

\section{Quantum bound of compiled biparite XOR games}
\label{sec:qb-xor}

In Section~\ref{sec:xor}, we established key insights into XOR games. Here, building upon the methodology outlined in~\cite{FOCS:NatZha23} and revisited in~\ref{cryptoNPA}, we present our first important result: compiled XOR games exhibit a quantum bias that closely aligns with the quantum bias of the corresponding nonlocal game, fluctuating only slightly with the security parameter. The power of the quantum prover to win a compiled XOR game with a probability larger than the optimal quantum winning probability for the corresponding XOR game crucially depends on their ability to transmit from the first round information about the received plaintext input. However, inputs are encrypted in such a way that the encryption satisfies IND-CPA security, meaning that even the quantum prover cannot do better than randomly guessing its question $x$ knowing the encryption $\xe$ of $x$. This inability of a QPT prover to break the encryption is articulated in Lemma \ref{lem:qhe-bound}.
In essence, this lemma conveys that regardless of the measurement employed by a QPT prover in the second round of the game, they are unable to differentiate between states resulting from distinct samples of plaintext questions taken in the first round.

Before stating our main result let us state two lemmas about the behavior of the XOR game operator under the pseudo-expectation map defined in Sec.~\ref{cryptoNPA}. The first lemma is a generalization of Claims 31 and 33 from~\cite{FOCS:NatZha23}.
\begin{lem}\label{lemmasquare}
    Let $P_x = A_x - \sum_{y}F_{xy}B_y$ with$\{A_x\}_{x\in \X}$ and $\{B_y\}_{y\in \Y}$  binary observables. Then there exists a  negligible function $\delta_{\mathrm{QHE}}(\cdot)$ such that 
    we have $\Tilde{\mathds{E}}\left[P_x^\dagger P_x
    \right] \geq - \delta_{\mathrm{QHE}}(\kappa)$, where $\Tilde{\mathds{E}}[\cdot]$ is the pseudo-expectation map defined in Sec.~\ref{cryptoNPA}.
\end{lem}
\begin{proof}
    Let us introduce the shortened notation $\hat{B}_x = \sum_yF_{xy}B_y$.

\begin{align*}
\Tilde{\mathds{E}}\left[P_x^\dagger P_x\right]
&= \Tilde{\mathds{E}}\left[A_x^\dagger A_x\right] - \Tilde{\mathds{E}}\left[A_x^\dagger\hat{B}_x\right] - \Tilde{\mathds{E}}\left[\hat{B}_x^\dagger A_x \right] + \Tilde{\mathds{E}}\left[\hat{B}_x^\dagger \hat{B}_x\right]\\
&=1 - 2\sum_yF_{xy}C_{xy} + \sum_{y,y'}F_{yx}F_{xy'}S_{yy'}\\
&=1 -  2\sum_yF_{xy}\mathop{\mathds{E}}_{\xe=\enc(x)} \sum_\aen (-1)^{\dec(\aen)}\bra{\psi} \M_{\aen|\xe}^{\dagger} B_y \M_{\aen|\xe}\ket{\psi} \\
&\qquad + \sum_{y,y'}F_{yx}F_{xy'}\mathop{\mathds{E}}_{z\in \X} \mathop{\mathds{E}}_{\xe=\enc(z)} \sum_\aen \bra{\psi} \M_{\aen|\xe}^{\dagger} B_y B_{y'} \M_{\aen|\xe}\ket{\psi}\\
&=1 -  2\mathop{\mathds{E}}_{\xe=\enc(x)} \sum_\aen (-1)^{\dec(\aen)}\bra{\psi} \M_{\aen|\xe}^{\dagger} \hat{B}_x \M_{\aen|\xe}\ket{\psi} + \mathop{\mathds{E}}_{z\in \X} \mathop{\mathds{E}}_{\xe=\enc(z)} \sum_\aen \bra{\psi} \M_{\aen|\xe}^{\dagger} \hat{B}_x^\dagger \hat{B}_{x} \M_{\aen|\xe}\ket{\psi}\\
&\approx_{\delta_{\mathrm{QHE}}(\kappa)} 1 -  2\mathop{\mathds{E}}_{\xe=\enc(x)} \sum_\aen (-1)^{\dec(\aen)}\bra{\psi} \M_{\aen|\xe}^{\dagger} \hat{B}_x \M_{\aen|\xe}\ket{\psi} +  \mathop{\mathds{E}}_{\xe=\enc(x)} \sum_\aen \bra{\psi} \M_{\aen|\xe}^{\dagger} \hat{B}_x^\dagger \hat{B}_{x} \M_{\aen|\xe}\ket{\psi}\\
&= \mathop{\mathds{E}}_{\xe=\enc(x)} \sum_\aen \bra{\psi} \M_{\aen|\xe}^{\dagger} \left( 
\mathds{1} -2 (-1)^{\dec(\aen)}  \hat{B}_x + \hat{B}_x^2
\right)\M_{\aen|\xe} \ket{\psi}\\
&= \mathop{\mathds{E}}_{\xe=\enc(x)} \sum_\aen \bra{\psi} \M_{\aen|\xe}^{\dagger} 
\left( 
(-1)^{\dec(\aen)} \mathds{1} - \hat{B}_x
\right)^2
\M_{\aen|\xe}\ket{\psi}
    \end{align*}
In the first two lines, we used the linearity and the definition of the pseudo-expectation map. In the third line, we used the definition of matrices $C$ (eq.~\eqref{pseudoExpAB}) and $S$ (eq.~\eqref{pseudoExpAABB}). In the fourth line, we just used the definition of $\hat{B}_x$. To get the fifth line we used the fact that $\hat{B}_x^\dagger\hat{B}_x$ has a QPT-implementable block encoding (Corrolaries~\ref{cor:add-block-encoding} and~\ref{cor:product-block-encoding}) and thus we can apply Lemma~\ref{lem:qhe-bound} given above to fix the  input $z \in \X$ to be exactly $x$, paying the price of a $\delta_{\mathrm{QHE}}(\kappa)$.
To get the sixth line we used linearity and we grouped all terms inside of one average over $\xe=\enc(x)$; in particular we noticed that we could write the first term $1$ as $\mathds{1}$ inside this average.
Finally, in the seventh line we recognised that the expression within round brackets can be expressed as the square of an operator.
This is necessarely positive, implying that up to a negligible factor $\Tilde{\mathds{E}}\left[P_x^\dagger P_x\right]$  must also be positive.
\end{proof}

\begin{remark}
This lemma is the key-point that allows to generalize the framework introduced in \cite{FOCS:NatZha23} to all games for which we have a SOS decomposition whose polynomials can be written as  $P_x = A_x - \sum_{y}F_{xy}B_y$. This specific shape might seem restricting, but we showed in the previous section that it is expressive enough to describe all XOR games. 
Extending the statement to more general forms of the polynomials might be challenging, because we lack of a clear intuition on how to interpret product of Alice's observables in the compiled scenario. The pseudo-expectation map is designed to mimic a true expectation map, assigning non-negative values to positive operators. Its small deviation from positivity, allowed only by a negligible function of the security parameter, is justified by the IND-CPA security of the encryption scheme. However, this justification relies on the SOS decomposition having the specific form described above. More general SOS forms lead to pseudo-expectation values that, with current techniques, cannot be meaningfully estimated.
\end{remark}

\begin{lem}\label{lemmaBpos}
    Let $P$ be a positive semi-definite homogeneous degree-2 polynomial over binary observables $\{B_y\}_y$. Then $\Tilde{\mathds{E}}\left[P
    \right] \geq 0$, where $\Tilde{\mathds{E}}[\cdot]$ is the pseudo-expectation map defined in Sec.~\ref{cryptoNPA}.
\end{lem}
\begin{proof}
    The pseudo-expectation of $P$ reads
    \begin{align}
        \Tilde{\mathds{E}} \left[ P\right] = 
    \mathop{\mathds{E}}_{x \in \X} \mathop{\mathds{E}}_{\xe: \dec(\xe)=x} \sum_\aen  \bra{\psi} \M_{\aen|\xe}^{\dagger} P \M_{\aen|\xe} \ket{\psi} \geq 0.
    \end{align}
    Since all of the elements in the sum are positive, because $P$ is positive semi-definite, the relation trivially holds.
\end{proof}

We can now state the main theorem, providing the quantum optimal bias of compiled XOR games.

\begin{thm}\label{thm:qbound_XOR}
    Given an XOR game with optimal quantum bias $\xi_q$, the optimal quantum bias of the compiled XOR game is $\xi_q+\delta_{\mathrm{QHE}}(\kappa)$, where $\delta_{\mathrm{QHE}}(\cdot)$ is a negligible function.
\end{thm}
\begin{proof}
    Let us recall Theorem~\ref{sosThm} which states that for every XOR game we have
    \begin{equation*}
    \xi_q\idd - \mathcal{B}_g = \sum_x\frac{\lambda_{x}}{2}\left(A_x - \sum_{y}F_{xy}B_y\right)^2 + P\left(\{B_y\}_y\right),
\end{equation*}

By applying the pseudo-expectation map we get
\begin{align}
\Tilde{\mathds{E}}\left[\xi_q\idd - \mathcal{B}_g\right] &= \Tilde{\mathds{E}}\left[\sum_x\frac{\lambda_{x}}{2}\left(A_x - \sum_{y}F_{xy}B_y\right)^2 + P\left(\{B_y\}_y\right)\right]\\
    & = \sum_x\frac{\lambda_{x}}{2}\Tilde{\mathds{E}}\left[\left(A_x - \sum_{y}F_{xy}B_y\right)^2\right] + \Tilde{\mathds{E}}\left[P\left(\{B_y\}_y\right)\right].
\end{align}
To get the second line we used the linearity of the pseudo-expectation map. Based on Lemmas~\ref{lemmasquare} implies that the first term on the r.h.s is up to a negligible function nonnegative, while Lemma~\ref{lemmaBpos} implies that the second term on the r.h.s. must be nonnegative.  The pseudo-expectation map on the l.h.s. gives $\xi_q - \bar{\xi}_q$, implying that in the worst case the compiled game optimal bias can be only negligibly larger than the optimal quantum bias of the corresponding XOR game: 
\begin{equation}
    \bar{\xi}_q \leq \xi_q + \delta_{\mathrm{QHE}}(\kappa),
\end{equation}
which completes the proof.
\end{proof}

\section{Quantum bound of compiled nonlocal games with many outputs and conditions for self-testing of qudit measurements from a single prover}\label{sec:st_qudits}

We now introduce a family of nonlocal games in which the players respond with more than two outputs. In this case, the quantum strategy involves $d$-outcome measurements, $\{M_{a|x}\}_{a=0}^{d-1}$ for Alice and $\{N_{b|y}\}_{b=0}^{d-1}$ for Bob. These measurements can be represented through generalized observables
\begin{equation}
    A_x^{(k)} = \sum_{a=0}^{d-1} \omega^{ak} M_{a|x},   \qquad B_y^{(l)} = \sum_{b=0}^{d-1} \omega^{bl} N_{b|y}, \quad   \forall k,l = 0,\cdots,d-1, 
\end{equation}
where $\omega=\exp \left(\frac{2 \pi i}{d}\right)$. Clearly, the eigenvalues of the generalized unitary observables $A_x^{(k)}$ and $B_y^{(l)}$ are roots of unity. We will refer to the indices $k$ and $l$ as the orders of the generalised measurement; notice that the orders $k=0 \mod d$ are the identity. For $d>2$, generalized observables are not Hermitian, but still satisfy the following regular property:
\begin{equation*}
    (A^{(k)})^\dagger = A^{(-k)} \qquad  (B^{(l)})^\dagger = B^{(-l)}
\end{equation*}
and hence  $\left[A^{(k)}\right]^\dagger A^{(k)} = \idd$ for every $k$.
Considering only projective measurements, it is easy to check that $A^{(k)} = \left[A^{(1)}\right]^k$ and $B^{(l)} = \left[B^{(1)}\right]$; we will use these two notations interchangeably.

In this section, we adopt the perspective of Bell inequalities.
The CHSH inequality is the simplest Bell inequality whose maximal violation self-tests the maximally entangled pair of qubits. Several propositions have been formulated to capture different features of the CHSH inequality, some focusing on self-testing maximally entangled shared states~\cite{SATWAP17}, some on self-testing mutually unbiased bases measurements for both parties~\cite{kaniewski2019maximal}. We concentrate on the Salavrakos-Augusiak-Tura-Wittek-Acin-Pironio (SATWAP) Bell inequality introduced in \cite{SATWAP17}, for which a self-test of arbitrary local dimension with minimal number of measurements ($2$ per part) has been proposed in \cite{SSKA21}.
We refer to this inequality as $d$-CHSH, because it shares many similarities with the CHSH inequality, generalising it to outputs that are $d$-dimensional.
However, this is not related to the class of Bell inequalities introduced  by Buhrman and Massar in~\cite{BM05}, and later developed in \cite{bavarian2014}, \cite{Liang09} and \cite{Ram16}, which are also known as $d$-CHSH inequalities.

The SATWAP inequality is defined for two spatially separated players, each receiving a bit as input, $\X = \Y = \{1,2\}$, and subsequently producing an output that can take one of $d$ different values, $\A = \B =  \{0,\dots,d-1\}$.
Many Bell inequalities can be expressed by conveniently defining some regular combination of probabilities given fixed inputs, called correlators.
The definition of the generalised correlator is:
\begin{equation}\label{eq:generalised_correlator}
    \langle A_x^{(k)} B_y^{(l)} \rangle = 
    \sum_{a=0}^{d-1} \sum_{b=0}^{d-1} \omega^{ak+bl} p(ab|xy).
\end{equation}
The SATWAP inequality is given in terms of these generalised correlators:
\begin{align}\label{eq:SATWAP_B}
    \beta_d^{\text{SATWAP}} =
    \sum_{k=1}^{d-1} \left( 
    a_k \langle A_1^k B_1^{d-k} \rangle +
    a_k^* \omega^k \langle A_1^k B_2^{d-k} \rangle+
    a_k^* \langle A_2^k B_1^{d-k} \rangle+
    a_k \langle A_2^k B_2^{d-k} \rangle 
    \right),
\end{align}
 with the following definition of the phases $a_k = \frac{\omega^{\frac{2k-d}{8}}}{\sqrt{2}} = \frac{1 - i}{2} \omega^{k/4}$  such that  $a_k^* = a_{d-k}$ and $a_{\pm d} = \pm \frac{1+i}{2}$.
To simplify the notation it is convenient to group Bob's observables in the following sums:
 \begin{equation}\label{c12}
      C_1^{(k)} = a_k B_1^{-k} + a_k^* \omega^k B_2^{-k},  \qquad
      C_2^{(k)} = a_k^* B_1^{-k} + a_k B_2^{-k}. 
 \end{equation}
Using these definitions, the corresponding Bell operator is thus
\begin{equation}
    \mathcal{B}_d^{\text{SATWAP}} =  \sum_{k=1}^{d-1}
    \left(
     A_1^k\otimes C_1^{(k)}  +  A_2^k\otimes C_2^{(k)} \right).
\end{equation}
Bell value is thus given as $\beta_d^{\text{SATWAP}} = \bra{\psi}\mathcal{B}_d^{\text{SATWAP}}\ket{\psi}$. The classical and the quantum bounds of this inequality are 
\begin{equation*}
    \left(\beta_d^{\text{SATWAP}}\right)_l = \frac{1}{2} \left(3 \cot\left(\frac{\pi}{4}d\right) - \cot\left(\frac{3\pi}{4}d\right)\right) - 2,
    \qquad  
    \left(\beta_d^{\text{SATWAP}}\right)_q = 2(d-1).
\end{equation*}

The quantum bound can be found by explicitly building the SOS decomposition of the shifted Bell operator, as shown in~\cite{SATWAP17}. The terms of the SOS decomposition are labeled by $x \in \{1,2\}$ and $k \in \{0, \dots, d-1\}$:
\begin{align*}
    P_{x,k} &= (A_x^k)^\dagger - C_x^{(k)}.
\end{align*}
A quick calculation allows us to check the correctness of the SOS decomposition :
\begin{align}\label{shiftedSATWAP}
    \beta_q \mathds{1} - \mathcal{B}_d^{\text{SATWAP}} &= 
    \frac{1}{2} \sum_{k=1}^{d-1} \left( 
    P_{1,k}^\dagger P_{1,k} + P_{2,k}^\dagger P_{2,k}
    \right) \\ \nonumber
    &= - \mathcal{B}_d^{\text{SATWAP}} +(d-1) \mathds{1} + \frac{1}{2} \sum_{k=1}^{d-1} \left( C_1^{(d-k)} C_1^{(k)} +C_2^{(d-k)} C_2^{(k)} \right) \\ \nonumber
    &= - \mathcal{B}_d^{\text{SATWAP}} +2(d-1).
\end{align}
The maximal quantum violation of this inequality is a self-test of a maximally entangled pair of qudits, together with measurements that are unitarily equivalent to the measurements used to maximally violate well-known Collins-Gisin-Linden-Massar-Popescu (CGLMP) Bell inequalities~\cite{collins2002bell}. 
We summarise this technical result in the following lemma, proven in~\cite{SSKA21}.
\begin{lem}[\cite{SSKA21}]\label{sska21}
The maximal violation of the SATWAP Bell inequality certifies the following:
\begin{itemize}
    \item The dimension of Alice's and Bob's Hilbert spaces is a multiple of $d$, and we can write their Hilbert space as the tensor products $\mathcal{H}_A = \mathds{C}^d \otimes \mathcal{H}_{A'}$, $\mathcal{H}_B = \mathds{C}^d \otimes \mathcal{H}_{B'}$,  where $\mathcal{H}_{A'}$ and $\mathcal{H}_{B'}$ are auxiliary Hilbert spaces of finite dimension;
    \item There exist local unitary transformations $U_A : \mathcal{H}_A \to \mathcal{H}_A$ and $U_B : \mathcal{H}_B \to \mathcal{H}_B$ such that
    \begin{align}\label{dBobs}
        & U_B B_1 U_B^\dagger = Z_d \otimes \mathds{1}_{B'}, && U_B B_2 U_B^\dagger = T_d \otimes \mathds{1}_{B'},\\
        & U_A A_1 U_A^\dagger = (a_1^* Z_d + 2 (a_1^*)^3 T_d) \otimes \mathds{1}_{A'}, && U_A A_2 U_A^\dagger = (a_1 Z_d  - a_1^* T_d) \otimes \mathds{1}_{B'}.
    \end{align}
    where $Z_d = \text{diag}[1, \omega, \dots, \omega^{d-1}]$ and $ T_d = \sum_{i=0}^{d-1} \omega^{i + 1/2} \ket{i}\bra{i} - \frac{2}{d} \sum_{j,i=0}^{d-1} (-1)^{\delta_{i,0} + \delta_{j,0}} \omega^{(i + j + 1)/2} \ket{i}\bra{j}$.
    \item Alice and Bob share a state $\ket{\psi_{AB}}$ which is unitarily equivalent to the maximally entangled pair of qudits
    \begin{equation*}
        U_A \otimes U_B \ket{\psi_{AB}} = \ket{\phi^+_d} \otimes \ket{\tau_{A'B'}}.
    \end{equation*}
\end{itemize}
\end{lem}
Note that Eqs.~\eqref{dBobs} fully characterize Bob's measurements, given that they are projective and $B_y^{(k)} = B_y^k$.

For the compiled SATWAP Bell inequality, a modification in the modeling of the quantum prover is necessary. Our focus, thus far, has primarily been on compiled games with binary outputs. The prover is initiated in some quantum state $\ket{\psi}$ and its action in the first round is described in the same way, with eqs.~\eqref{eq:postMstate} and~\eqref{eq:probenc} still holding. However, a generalized $k$-th order decrypted observable now reads:
\begin{align}
    \rA_x^{(k)} &= \mathop{\mathds{E}}_{\xe=\enc(x)}\sum_{\aen} \omega^{k \dec(\aen)}\M_{\aen|\xe}^\dagger\M_{\aen|\xe}\\
    &= \mathop{\mathds{E}}_{\xe=\enc(x)}A_{\xe}^{(k)}
\end{align}
and the generalized decrypted correlator takes the form
\begin{equation}\label{eq:cryptocorr_gen}
     \langle\rA_x^{(k)},B_y^{(l)}\rangle = \mathop{\mathds{E}}_{\xe:\dec(\xe) = x}\sum_{\aen}\omega^{k\dec(\aen)}\bra{\psi_{\aen|\xe}}B_y^{(l)}\ket{\psi_{\aen|\xe}},
\end{equation}
with this value operationally giving the expectation value of $\omega^{ak+bl}$, which will be directly used to asses the performance of the prover in the SATWAP game. 

We generalize the modified moment matrix $\tilde{\mathcal{Q}}$ to include the expectation values of generalized observables.
Recall that generalized observables are labeled not only by the input $x \in \X$ (or $y \in \Y$) but also by the degree $k \in \{1,\dots,  d-1\}$ (or $l$).
Every entry of the modified covariance matrix $\tilde{\mathcal{Q}}$ is labeled by two generalized observables, hence is a square matrix of dimension $(|\X|+|\Y|) (d-1)$. Hence, we define a generalization of the matrix $\tilde{\mathcal{Q}}$ as follows
\begin{equation}
    \tilde{\mathcal{Q}}=\left[
\begin{array}{c|c}
\idd_{(d-1)|\X|}  & {C} \\ \hline
 {C}^T & {S}
\end{array}\right]
\end{equation}
where $\idd_{(d-1)|\X|}$ is  $(d-1)|\X|\times (d-1)|\X|$ identity matrix and 
\begin{align}
    {C} &= \sum_{x,y,k,l}\langle \rA_x^{(d-k)},B_y^{(l)}\rangle \ketbra{(k-1)d + x}{(l-1)d+y}\\
    {S} &= \mathop{\mathds{E}}_{x \in \X} \mathop{\mathds{E}}_{\xe: \dec(\xe)=x} \sum_\aen  \bra{\psi_{\aen|\xe}}  B_y^{(d-l)} B_{y'}^{(l')}  \ket{\psi_{\aen|\xe}}\ket{(l-1)d + y}\bra{(l'-1)d +y'}
\end{align}
with $\langle \rA_x^{(k)},B_y^{(l)}\rangle$ being introduced in eq.~\eqref{eq:cryptocorr_gen}. 

As in~\ref{cryptoNPA}, we define a linear operator $\Tilde{\mathds{E}}$ that maps every homogeneous degree-2 polynomial in the variables $\{A_x^{(k)},B_y^{(l)}\}_{x,y,k,l}$  to linear combinations of elements of matrix $\tilde{\mathcal{Q}}$ in the following way 
\begin{align}\label{pseudoExpABkl}
    &\Tilde{\mathds{E}}[A_x^{(d-k)}B_y^{(l)}] = {C}_{(k-1)d + x,(l-1)d+y},&&  \Tilde{\mathds{E}}[B_y^{(l)}\rA_x^{(d-k)}] = {C}^{T}_{(l-1)d+y,(k-1)d + x}\\ \label{pseudoExpAABBkkll}
    &\Tilde{\mathds{E}}[A_x^{(k)}A_{x'}^{(k')}] = \delta_{x,x'}\delta_{k,k'},&&  \Tilde{\mathds{E}}[B_y^{(l)}B_{y'}^{(l')}] = S_{(l-1)d + y,(l'-1)d+y'}
\end{align}
For $y=y'$, $l=l'$ or $x=x'$, $k=k'$ in eq.~\eqref{pseudoExpAABB} we get a consistent mapping of identity
\begin{equation}\label{xxyykkll}
    \Tilde{\mathds{E}}[A_x^{(d-k)}A_{x}^{(k)}] =  \Tilde{\mathds{E}}[\idd] = 1,\qquad  \Tilde{\mathds{E}}[B_y^{(d-l)}B_y^{(l)}] = \Tilde{\mathds{E}}[\idd] = S_{(l-1)d + y,(l-1)d+y} = 1.
\end{equation}

With this formalism, we can prove a result equivalent to Theorem~\eqref{thm:qbound_XOR} for the quantum bound of the compiled SATWAP inequality. Let us first introduce the following lemma.

\begin{lem}\label{lem:PP_SATWAP}
    Let $P_{x,k} = (A_x^k)^\dagger - C_x^{(k)}$, where $A_x^k$ is a generalised observable, and $C_x^{(k)}$ are linear sums of generalised observables defined in eq.~\eqref{c12}. Then there exists a negligible function $\delta_{\mathrm{QHE}}(\cdot)$ such that we have $\Tilde{\mathds{E}}\left[(P_x^k)^\dagger P_x^k
    \right] \geq -\delta_{\mathrm{QHE}}(\kappa)$, where $\Tilde{\mathds{E}}[\cdot]$ is the pseudo-expectation map generalised above.
\end{lem}
\begin{proof}
    The proof follows closely the arguments presented in the proof of Lemma~\ref{lemmasquare}, with the difference that generalised observables are not hermitian, and phases are also complex.
\begin{align*}
        \Tilde{\mathds{E}}\left[P_{x,k}^\dagger P_{x,k}\right]
        &= \Tilde{\mathds{E}}\left[(A_x^{(k)})^\dagger A_x^{(k)}\right] - \Tilde{\mathds{E}}\left[(A_x^{(k)})^\dagger C_x^{(k)}\right] - \Tilde{\mathds{E}}\left[(C_x^{(k)})^\dagger A_x^{k}\right] + \Tilde{\mathds{E}}\left[(C_x^{(k)})^\dagger C_x^{(k)}\right]\\
        &= \Tilde{\mathds{E}}\left[A_x^{(d-k)} A_x^{(k)}\right] 
        - \Tilde{\mathds{E}}\left[A_x^{(d-k)} C_x^{(k)}\right] 
        - \Tilde{\mathds{E}}\left[C_x^{(d-k)} A_x^{(k)}\right] 
        + \Tilde{\mathds{E}}\left[C_x^{(d-k)} C_x^{(k)}\right]\\
        &=1 
        - \mathop{\mathds{E}}_{\xe=\enc(x)} \sum_\aen \omega^{(d-k)\dec(\aen)}\bra{\psi_{\aen|\xe}} C_x^{(k)} \ket{\psi_{\aen|\xe}} 
        - \mathop{\mathds{E}}_{\xe=\enc(x)} \sum_\aen \omega^{k\dec(\aen)}\bra{\psi_{\aen|\xe}} C_x^{(d-k)} \ket{\psi_{\aen|\xe}} \\
        &+ \mathop{\mathds{E}}_{i\in \X} \mathop{\mathds{E}}_{\xe=\enc(i)} \sum_\aen \bra{\psi_{\aen|\xe}} C_x^{(d-k)} C_x^{(k)} \ket{\psi_{\aen|\xe}}\\
        &\approx_{\delta_{\mathrm{QHE}}(\kappa)} \mathop{\mathds{E}}_{\xe=\enc(x)} \sum_\aen \bra{\psi_{\aen|\xe}}
        \left(  \mathds{1} 
            - \omega^{-k \dec(\aen)} C_{x}^{(k)} 
            - \omega^{k \dec(\aen)}  C_{x}^{(d-k)}
            + C_{x}^{(d-k)} C_x^{(k)} \right)
        \ket{\psi_{\aen|\xe}}\\
        &\approx_{\delta_{\mathrm{QHE}}(\kappa)} \mathop{\mathds{E}}_{\xe=\enc(x)} \sum_\aen \bra{\psi_{\aen|\xe}}
        \left( 
        \omega^{k\dec(\aen)} \mathds{1} - C_{x}^{(k)}
        \right)^2
        \ket{\psi_{\aen|\xe}} \geq 0
\end{align*}
In the second and in the third line we used the linearity of the pseudo-expectation map over sums of generalised observables and the definitions stated above.
In the fourth step we apply Lemma~\ref{lem:qhe-bound} to the last addend and we fix the index $i=x$, at the price of a negligible function $\delta_{\mathrm{QHF}}(\kappa)$. We can apply this lemma because $C_{x}^{(d-k)} C_x^{(k)}$ has a QPT-implementable block encoding, since is composed by linear sums and multiplications of QPT-implementable generalised observables (Corollaries~\ref{cor:add-block-encoding} and~\ref{cor:product-block-encoding}).
Finally, using linearity we can group all the terms inside of one average over $\xe=\enc(x)$; we rephrase this sum as an operator square, which is non-negative by definition.
\end{proof}
\begin{thm}
    Let's consider the $d$-dimensional SATWAP Bell inequality, with quantum bound $(\beta_d^{SATWAP})_q$, and its compiled version through Kalai protocol.
    If $d$ is polynomial w.r.t. the security parameter $\kappa$, then the quantum bound of the compiled SATWAP inequality is $(\beta_d^{SATWAP})_q + \theta(\kappa)$, where $\theta(\cdot)$ is a negligible function.
\end{thm}
\begin{proof}
Applying the map $\Tilde{\mathds{E}}$ to the shifted SATWAP operator given in eq.~\eqref{shiftedSATWAP} gives :
\begin{align}\label{pseudoexpectation_shiftedSATWAP}
    \Tilde{\mathds{E}}\left[\beta_q \mathds{1} - \mathcal{B}_d^{\text{SATWAP}}\right] &= 
    \frac{1}{2} \sum_{k=1}^{d-1} \left( 
\Tilde{\mathds{E}}\left[P_{1,k}^\dagger P_{1,k}\right] +  \Tilde{\mathds{E}}\left[P_{2,k}^\dagger P_{2,k}\right]
    \right),
\end{align}
In Lemma~\ref{lem:PP_SATWAP} we proved that all elements in the sum on the r.h.s. are non-negative up to a negligible function, \emph{i.e.}
\begin{equation}\label{jedna}
\Tilde{\mathds{E}}\left[P_{x,k}^\dagger P_{x,k}\right] \approx_{\delta} \mathop{\mathds{E}}_{\xe=\enc(x)} \sum_\aen \bra{\psi} \M_{\aen|\xe}^{\dagger} \left(\omega^{k\dec(\aen)}\idd - C_{x}^{(k)}\right)^2 \M_{\aen|\xe}\ket{\psi}
\end{equation}
This implies that there exists a negligible function of the security parameter $\theta(\kappa)$ such that the maximal quantum score $\bar{\beta}_q$ in the compiled SATWAP inequality becomes
\begin{equation}
    \bar{\beta}_q \leq 2(d-1) + \theta(\kappa).
\end{equation}
\end{proof}

Let us now explore self-testing properties, and assume that a QPT prover achieved the score $\bar{\beta}_q$ in the compiled SATWAP inequality.
Given eq.~\eqref{pseudoexpectation_shiftedSATWAP}, reaching the optimal quantum score implies
\begin{equation}\label{zero}
0 = \frac{1}{2} \sum_{k=1}^{d-1} \left( 
\Tilde{\mathds{E}}\left[P_{1,k}^\dagger P_{1,k}\right] +  \Tilde{\mathds{E}}\left[P_{2,k}^\dagger P_{2,k}\right]
    \right).
\end{equation}
By using eqs.~\eqref{jedna} and~\eqref{zero} we get
\begin{equation}
    \mathop{\mathds{E}}_{\xe=\enc(x)} \sum_\aen \bra{\psi} \M_{\aen|\xe}^{\dagger} \left(\omega^{k\dec(\aen)}\idd - C_{x}^{(k)}\right)^2 \M_{\aen|\xe}\ket{\psi} \approx_{\delta} 0,
\end{equation}
where we write $\delta$ as shorthand for $\delta_\mathrm{QHE}$.
This further gives us
\begin{equation}
    \mathop{\mathds{E}}_{\xe=\enc(x)} \sum_\aen  \left\|C_x^{(k)}\M_{\aen|\xe}\ket{\psi} - \omega^{k\dec(\aen)}\M_{\aen|\xe}\ket{\psi}\right\|^2 \approx_{\delta} 0,
\end{equation}
further implying that for every $k$, every $\xe$ decrypting to $x$, and every $\aen$
\begin{equation}\label{ca}
     C_x^{(k)}\M_{\aen|\xe}\ket{\psi} \approx_{\delta} \omega^{k\dec(\aen)}\M_{\aen|\xe}\ket{\psi}.
\end{equation}
If we fix $k'$ the previous relation implies
\begin{align*}
     \left[C_x^{(1)}\right]^{k'}\M_{\aen|\xe}\ket{\psi} &\approx_{\delta} \left[C_x^{(1)}\right]^{k'-1}C_x^{(1)}\M_{\aen|\xe}\ket{\psi}\\
     &\approx_{2\cdot \delta} \omega^{\dec(\aen)}\left[C_x^{(1)}\right]^{k'-1}\M_{\aen|\xe}\ket{\psi}\\
     &\approx_{3\cdot \delta} \omega^{2\dec(\aen)}\left[C_x^{(1)}\right]^{k'-2}\M_{\aen|\xe}\ket{\psi}\\
     &\cdots\\
     &\approx_{k' \cdot \delta} \omega^{k'\dec(\aen)}\M_{\aen|\xe}\ket{\psi}\\
     &\approx_{(k'+1) \cdot \delta} C_x^{(k')}\M_{\aen|\xe}\ket{\psi},
\end{align*}
where in the second line we used eq.~\eqref{ca} for $k=1$, which we used successively until the last line where we used again eq.~\eqref{ca} for $k = k'$.
By summing over $\aen$, and noting that $k' < d$ which is polynomially bounded in the security parameter $\kappa$ we get:
\begin{equation}
    \left[C_x^{(1)}\right]^{k}\ket{\psi} \approx_{O(\delta)}  C_x^{(k)}\ket{\psi}.
\end{equation}
Now we develop the following expression
for an arbitrary $k$
\begin{align*}
    C_x^{(d-k)}C_x^{(k)}\M_{\aen|\xe}\ket{\psi} &\approx_\delta\omega^{k\dec(\aen)}C_x^{(d-k)}\M_{\aen|\xe}\ket{\psi}\\
    &\approx_{2\delta} \omega^{k\dec(\aen)}\omega^{(d-k)\dec(\aen)}\M_{\aen|\xe}\ket{\psi}\\
    &=\omega^{d\dec(\aen)}\M_{\aen|\xe}\ket{\psi}\\
&= \M_{\aen|\xe}\ket{\psi},
\end{align*}
where for the first line we used eq.~\eqref{ca}, for the second line we used the same equation by changing $k$ to $d-k$, and the last line follows from $\omega^d = 1$. Again by summing over $\aen$ we obtain
\begin{equation}
    C_x^{(d-k)}C_x^{(k)}\ket{\psi} \approx_{O(\delta)} \ket{\psi},
\end{equation}
Thus, on the support of $\ket{\psi}$ the observables $C_x^{(k)}$  satisfy condition
\begin{equation}
    \left[C_x^{(1)}\right]^{k} \approx_{O(\delta)} C_x^{(k)},\qquad C_x^{(d-k)}C_x^{(k)} \approx_{O(\delta)} \idd.
\end{equation}
In~\cite{SSKA21} it is proven that these two equations for $O(\delta) = 0$ imply though Lemmas 1, 2, and 3 in the supplementary material therein that there exists a unitary $U$ such that
\begin{equation}\label{exactST}
    U_B B_1^{(1)} U_B^\dagger = Z_d \otimes \mathds{1}_{B'}, \qquad U_B B_2^{(1)} U_B^\dagger = T_d \otimes \mathds{1}_{B'},
\end{equation}
with $Z_d$ and $T_d$ being given in Lemma~\ref{sska21}. Since $B_1^{(1)}$ and $B_2^{(1)}$ are unitary we get $B_1^{(k)} = \left[B_1^{(1)}\right]^k$ and $B_2^{(k)} = \left[B_2^{(1)}\right]^k$.
We conjecture that the derivation procedure is robust to noise, implying that in the case of reaching the optimal quantum violation of the compiled SATWAP inequality relations~\eqref{exactST} hold up to a negligible function of the security parameter.

\section{Computational self-test of any two binary measurements from a single prover}\label{sec:st_2qubitmeas}

In the previous sections we showed that the KLVY compiler preserves the quantum bound for some classes of games; the natural follow-up question is to investigate whether the self-testing properties can also be translated.
This opens up the possibility of defining single-prover self-testing protocols, based on computational assumptions.

Let us start by focusing on the self-test of any pair of binary observable.
We identify in the non-locality literature an inequality that self tests such configuration; furthermore this is a specific XOR game, for which our result apply.
We explicitly carry out the calculations for the robust case as well, showing the feasibility of a single-prover self-test of any pair of qubit observables, with computational assumptions.

More precisely, consider the class of correlation Bell inequalities as introduced in \cite{Le_2023} and further discussed in \cite{barizien2023custom}. This family of Bell inequalities is characterized by three parameters, namely, $\mu$, $\nu$, and $\chi$. The Bell operator within this framework is defined as
\begin{align}\label{param}
\begin{split}
    \mathcal{B}_{\mu\nu\chi} &= \cos(\mu+\nu)\cos(\mu + \chi)(\cos(\chi)A_0-\cos(\nu)A_1)\otimes B_0 \\
    &\quad + \cos(\nu)\cos(\chi) (-\cos(\mu+\chi)A_0 +\cos(\mu+\nu)A_1)\otimes B_1.
    \end{split}
\end{align}
This formulation encompasses the CHSH inequality for $\mu=\chi=0$ and $\nu=\pi$. A Bell inequality belonging to this class exhibits quantum violations when $\cos(\mu+\chi)\cos(\mu+\nu)\cos(\nu)\cos(\chi) < 0$. In instances where this condition is met, achieving the quantum limit,
\begin{equation}\label{pg}
\beta_q = \pm \sin(\mu)\sin(\chi-\nu)\sin(\mu+\nu+\chi)
\end{equation}
self-tests the maximally entangled pair of qubits \cite{wooltorton2023deviceindependent}. This characteristic renders it a viable candidate for our compilation procedure. The sum-of-squares (SOS) decomposition of the shifted Bell operator is
\begin{align}\label{nk}
    \beta_q\idd - \mathcal{B}_{\mu\nu\chi}= c_0P_0^\dagger P_0 + c_1P_1^\dagger P_1,
\end{align}
where
\begin{equation*} 
c_0 = -\frac{\cos(\chi)\cos(\mu+\chi)}{2\sin(\mu)},\quad c_1 = -\frac{\cos(\nu)\cos(\mu+\nu)}{2\sin(\mu)}
\end{equation*}
and
\begin{align*}
    P_0 &= \sin(\mu)A_0 + \cos(\mu+\nu)B_0 -\cos(\nu)B_1\\
    P_1 &= \sin(\mu)A_1 + \cos(\mu+\chi)B_0 -\cos(\chi)B_1.
\end{align*}
Given the form of the SOS decomposition~\eqref{nk}, the optimal value of the compiled inequality $\bar{\beta}_q$ is the same as the quantum bound $\beta_q$ given in~\eqref{pg} up to a negligible function. Taking the pseudo-expectation value of Eq.~\eqref{nk} we get
\begin{align}
    \beta_q - \bar{\beta} = c_0\Tilde{\mathds{E}}\left[P_0^\dagger P_0\right] + c_1\Tilde{\mathds{E}}\left[P_1^\dagger P_1\right],
\end{align}
where $\bar{\beta}$ is the violation of the compiled inequality. As usual, we denote $\hat{B}_0 = (-\cos(\mu+\nu)B_0 +\cos(\nu)B_1)/\sin(\mu)$ and $\hat{B}_1 = -(\cos(\mu+\chi)B_0 +\cos(\chi)B_1)/\sin(\mu)$.
If $\bar{\beta} = \beta_q - \varepsilon$, the equation above simplifies to
\begin{equation}\label{wte}
    \varepsilon = c_0\Tilde{\mathds{E}}\left[P_0^\dagger P_0\right] + c_1\Tilde{\mathds{E}}\left[P_1^\dagger P_1\right]
\end{equation}
where, as outlined in Lemma~\ref{lemmasquare}, the two terms on the right-hand side differ from a nonnegative function only by a negligible function contingent upon the quantum homomorphic encryption scheme:
\begin{equation*}
     \Tilde{\mathds{E}}\left[P_x^\dagger P_x\right]
     \geq
     \mathop{\mathds{E}}_{\xe=\enc(x)} \sum_\aen  \left\|\hat{B}_x\ket{\psi_{\aen|\xe}} -(-1)^{\dec(\aen)}\ket{\psi_{\aen|\xe}}\right\|^2 - \delta_{\mathrm{QHE}}(\kappa).
\end{equation*}
Employing this bound in eq.~\eqref{wte} yields:
\begin{align*}
    \varepsilon  + (c_0 + c_1) \delta_{\mathrm{QHE}}(\kappa) \geq
    & c_0\mathop{\mathds{E}}_{\xe=\enc(0)} \sum_\aen  \left\|\hat{B}_0\ket{\psi_{\aen|\xe}} -(-1)^{\dec(\aen)}\ket{\psi_{\aen|\xe}}\right\|^2 \\
    +& c_1\mathop{\mathds{E}}_{\xe=\enc(1)} \sum_\aen  \left\|\hat{B}_1\ket{\psi_{\aen|\xe}} -(-1)^{\dec(\aen)}\ket{\psi_{\aen|\xe}}\right\|^2
\end{align*}
where now the right-hand side comprises solely positive quantities

Hence, for all values of $x$ we obtain:
\begin{align}\label{eq:SOSconstraint_robust}
    \mathop{\mathds{E}}_{\xe=\enc(x)} \sum_\aen  \left\|\hat{B}_x\ket{\psi_{\aen|\xe}} -(-1)^{\dec(\aen)}\ket{\psi_{\aen|\xe}}\right\|^2 \leq \frac{\varepsilon}{c_x} + \frac{c_0+c_1}{c_x}\delta_{\mathrm{QHE}}(\kappa)
    \qquad \forall x \in\{0,1\}.
\end{align}

We initially consider the noiseless and infinite security case \emph{i.e.} $\varepsilon = 0$ and $\delta_{\mathrm{QHE}} = 0$ for pedagogical purposes, acknowledging that the actual scenario does not conform to our simplified model. The equation above would imply that every addend is equal to zero, implying
\begin{align*}
   \hat{B}_x\ket{\psi_{\aen|\xe}} =(-1)^{\dec(\aen)}\ket{\psi_{\aen|\xe}}, \qquad \forall\xe:\dec(\xe) = x, \forall\aen.
\end{align*}
When the quantum bound is saturated, the square of the hat operator $\hat{B}_x$ acts like identity on $\ket{\psi_{\aen|\xe}}$:
\begin{align}\label{eq:Bhat_squared_id}
    (\hat{B}_x)^2 \ket{\psi_{\aen|\xe}} = (-1)^{\dec(\aen)} \hat{B}_x\ket{\psi_{\aen|\xe}} = \ket{\psi_{\aen|\xe}} \qquad \forall\xe:\dec(\xe) = x, \forall\aen.
\end{align}
By the definition of $\hat{B}_0$, and assuming $B_0^2 = B_1^2 = \mathds{1}$, the following is always true for every state:
\begin{align}\label{eq:Bhat_squared_comm}
    (\hat{B}_0)^2 \ket{\psi_{\aen|\xe}}
&=\left(\frac{\cos^2(\nu)+\cos^2(\nu+\mu)}{\sin^2(\mu)} \mathds{1} - \frac{\cos(\nu)\cos(\nu+\mu)}{\sin^2(\mu)}  \{B_0,B_1\} \right)\ket{\psi_{\aen|\xe}} \nonumber\\
    &= k_\mathds{1} \ket{\psi_{\aen|\xe}} - k_{\{\}}\{B_0,B_1\}\ket{\psi_{\aen|\xe}}
\end{align}
where in the second line we simply fixed a notation for the coefficients of the operators.
Equations \eqref{eq:Bhat_squared_id} and \eqref{eq:Bhat_squared_comm} together fix the value of the anti-commutator of Bob's observables
when the maximal quantum violation is achieved:
\begin{equation*}
    \{B_0,B_1\}\ket{\psi_{\aen|\xe}} 
    = \frac{k_\mathds{1}-1}{ k_{\{\}}} \ket{\psi_{\aen|\xe}}
    = 2 \cos(\mu)\ket{\psi_{\aen|\xe}} \qquad \forall\xe:\dec(\xe) = x, \forall\aen
\end{equation*}
where the second inequality is a simple trigonometric identity
\begin{align*}
    \frac{k_\mathds{1} - 1}{k_{\{\}}} = 
    \frac{\cos^2(\nu)+\cos^2(\nu+\mu) - \sin^2(\mu)}{\sin^2(\mu)}
    \frac{\sin^2(\mu)}{\cos(\nu)\cos(\nu+\mu)}
    = 2 \cos(\mu).
\end{align*}
Since all observables are binary, this is a self-test.
Indeed, Jordan's lemma ensures that $B_0$ and $B_1$ can be simultaneously block-diagonalised such that all blocks are either of size $2\times 2$ or $1\times 1$.  In the same way like in~\cite{Supi__2022}, we embed every $1\times 1$ into a Hilbert space of larger dimension. This operation does not affect the correlation probabilities, and it simplifies our analysis, as we work with a Jordan decomposition in which all blocks are of the size $2\times 2$. Without loss of generalisation, we can apply a local unitary to the observable $B_0$ and bring it to Pauli's $\sigma_x$ in every  $2\times 2$ block:
\begin{align}
    UB_0U^\dagger = \sigma_x\otimes\sum_{i}\proj{i}.
\end{align}
We can also safely bring all the $2\times 2$ blocks of $B_1$ to the $XZ$ plane of the Bloch sphere; by fixing the value of the anti-commutator, we get that $B_1$ has $\cos(\mu)\sigma_x + \sin(\mu)\sigma_z$ in every block. 
For every $\mu \neq 0$ we can find $\nu$ and $\chi$ such that the condition $\cos(\mu+\chi)\cos(\mu+\nu)\cos(\nu)\cos(\chi) < 0$ is satisfied, implying that the quantum bound is larger than the classical. This implies that in the noiseless and infinite QHE security case we can perform a single-prover self-test of any two observables applied by Bob, \emph{i.e.} for every $\mu$ there is a compiled XOR game whose maximal violation self-tests Bob's measurements implying:
\begin{equation}\label{exactST_binary}
    UB_0U^\dagger = \sigma_x \otimes \mathds{1}, \qquad UB_1U^\dagger = (\cos(\mu)\sigma_x+\sin(\mu)\sigma_z) \otimes \mathds{1}
\end{equation}

The assumptions made in our simplified case are unrealistic, both from a cryptographic and experimental point of view.
Assuming a finite security parameter $\kappa$ and small experimental deviations $\varepsilon$, we can still prove that the anti-commutator of Bob's observables is close to the noiseless value, obtaining a robust self-test for Bob's binary observables.
Let's start considering Eq.~\eqref{eq:SOSconstraint_robust} with $\xe=\enc(0)$, and open the sum over $\aen$:
\begin{align*}
    \mathop{\mathds{E}}_{\xe=\enc(0)}
    \left[
    \sum_{\aen=\enc(0)}  \left\| (\hat{B}_0 - \mathds{1})\ket{\psi_{\aen|\xe}}\right\|^2
    +\sum_{\aen=\enc(1)}  \left\| (\hat{B}_0 + \mathds{1})\ket{\psi_{\aen|\xe}}\right\|^2
    \right]
    \leq \frac{\varepsilon}{c_0} + \frac{c_0+c_1}{c_0}\delta_{\mathrm{QHE}}(\kappa)
\end{align*}
We call the vectors inside the norms as $\ket{\Delta_{\aen|\xe}^\pm} = (\hat{B}_0 \pm \mathds{1})\ket{\psi_{\aen|\xe}}$, hence we can write:
\begin{align}\label{eq:SOS_Delta}
    \mathop{\mathds{E}}_{\xe=\enc(0)}
    \left[\sum_{\aen=\enc(0)}  \left\| \ket{\Delta_{\aen|\xe}^-}\right\|^2
    +\sum_{\aen=\enc(1)}  \left\| \ket{\Delta_{\aen|\xe}^+}\right\|^2
    \right]\leq \frac{\varepsilon}{c_0} + \frac{c_0+c_1}{c_0}\delta_{\mathrm{QHE}}(\kappa)
\end{align}
The square of $\hat{B}_0$ is
\begin{align*}
    (\hat{B}_0)^2 
    &= \mathds{1} + (\hat{B}_0 - \mathds{1}) + \hat{B}_0 (\hat{B}_0 - \mathds{1})
    = \mathds{1} - (\hat{B}_0 + \mathds{1}) + \hat{B}_0 (\hat{B}_0 + \mathds{1})\\
    &= k_\mathds{1} \mathds{1} - k_{\{\}}\{B_0,B_1\}
\end{align*}
where the first lines are trivial identities, and in the second line we rewrote Eq.~\eqref{eq:Bhat_squared_comm}.
Rearranging the terms, and applying the observables on $\ket{\psi_{\aen|\xe}}$, we obtain the following formula for the anti-commutator of Bob's observables:
\begin{align*}
    \left[ 2 \cos(\mu) \mathds{1} - \{B_0,B_1\} \right] \ket{\psi_{\aen|\xe}} 
    &= \pm \frac{1}{ k_{\{\}} } \ket{\Delta_{\aen|\xe}^\mp}
    + \frac{1}{ k_{\{\}} } \hat{B}_0 \ket{\Delta_{\aen|\xe}^\mp}
\end{align*}
where to find $2\cos(\mu)$ we used again the trigonometric identity above.
Let us first focus on the equation with $\ket{\Delta_{\aen|\xe}^-}$.
We compute its norm squared, averaging over $\xe=\enc(0)$ and summing over $\aen =\enc(0)$
\begin{align*}
    \mathop{\mathds{E}}_{\xe=\enc(0)} \sum_{\aen=\enc(0)} 
    \left\| \left[ 2 \cos(\mu)\mathds{1} - \{B_0,B_1\} \right] \ket{\psi_{\aen|\xe}} \right\|^2 
    &= \frac{1}{ k_{\{\}}^2 } \mathop{\mathds{E}}_{\xe=\enc(0)} \sum_{\aen=\enc(0)} \left\|  \ket{\Delta_{\aen|\xe}^-} +  \hat{B}_0 \ket{\Delta_{\aen|\xe}^-}\right\|^2 \\
    &\leq \frac{1}{ k_{\{\}}^2 } \mathop{\mathds{E}}_{\xe=\enc(0)} \sum_{\aen=\enc(0)} 
    \left(\left\| \ket{\Delta_{\aen|\xe}^-}\right\|^2 +  \left\| \hat{B}_0 \ket{\Delta_{\aen|\xe}^-}\right\|^2\right)\\
    &\leq \frac{1}{ k_{\{\}}^2 } \mathop{\mathds{E}}_{\xe=\enc(0)} \sum_{\aen=\enc(0)} 
    \left(\left\| \ket{\Delta_{\aen|\xe}^-}\right\|^2 + \|\hat{B}_0 \|^2 \left\| \ket{\Delta_{\aen|\xe}^-}\right\|^2\right)\\
    &=  \frac{1 + \|\hat{B}_0 \|^2}{ k_{\{\}}^2 } \mathop{\mathds{E}}_{\xe=\enc(0)} \sum_{\aen=\enc(0)} 
    \left\| \ket{\Delta_{\aen|\xe}^-}\right\|^2
\end{align*}
where we used Cauchy-Schwarz and the triangle inequality.
Similarly, we can find a similar bound summing over $\aen =\enc(1)$ and using the decomposition of the anti-commutator with the vectors $\ket{\Delta_{\aen|\xe}^+}$:
\begin{align*}
    \mathop{\mathds{E}}_{\xe=\enc(0)} \sum_{\aen=\enc(1)} 
    \left\| \left[ 2 \cos(\mu)\mathds{1} - \{B_0,B_1\} \right] \ket{\psi_{\aen|\xe}} \right\|^2 
    &= \frac{1}{ k_{\{\}}^2 } \mathop{\mathds{E}}_{\xe=\enc(0)} \sum_{\aen=\enc(1)} \left\|  -\ket{\Delta_{\aen|\xe}^+} +  \hat{B}_0 \ket{\Delta_{\aen|\xe}^+}\right\|^2 \\
    &\leq  \frac{1 + \|\hat{B}_0 \|^2}{ k_{\{\}}^2 } \mathop{\mathds{E}}_{\xe=\enc(0)} \sum_{\aen=\enc(1)} 
    \left\| \ket{\Delta_{\aen|\xe}^+}\right\|^2
\end{align*}
Now, putting together these two result we can bound the sum over $\aen$
\begin{align*}
\mathop{\mathds{E}}_{\xe=\enc(0)} \sum_{\aen} 
    \left\| \left[ 2 \cos(\mu)\mathds{1} - \{B_0,B_1\} \right] \ket{\psi_{\aen|\xe}} \right\|^2 &\leq 
     \frac{1 + \|\hat{B}_0 \|^2}{ k_{\{\}}^2 } \mathop{\mathds{E}}_{\xe=\enc(0)}
    \left[
    \sum_{\aen=\enc(0)} \left\| \ket{\Delta_{\aen|\xe}^-}\right\|^2
    +\sum_{\aen=\enc(1)} \left\| \ket{\Delta_{\aen|\xe}^+}\right\|^2
    \right]\\
    &\leq \frac{1 + \|\hat{B}_0 \|^2}{ k_{\{\}}^2 }  \left( \frac{\varepsilon}{c_0} + \frac{c_0+c_1}{c_0}\delta_{\mathrm{QHE}}(\kappa)\right)\\
    &\leq \sin^2(\mu) \left[ 2 \cos(\mu)+\frac{1}{\cos(\nu)\cos(\nu+\mu)} \right]
    \left( \frac{\varepsilon}{c_0} + \frac{c_0+c_1}{c_0}\delta_{\mathrm{QHE}}(\kappa)\right)
\end{align*}
where in the second line we used the bound found in Eq.~\eqref{eq:SOS_Delta}.
The angles are fixed from the inequality, therefore in the last line we have some negligible functions in $\epsilon$ and $\kappa$ multiplied by some constant factors depending on the angles.
Therefore we proved that, on the support of the states used in the compiled game, the anti-commutator of Bob's observable is $2 \cos(\mu)\mathds{1}$ up to a function $\eta_{ anticomm }$:
\begin{equation}
     \mathop{\mathds{E}}_{\xe=\enc(0)} \sum_{\aen} 
    \left\| \left[ 2 \cos(\mu)\mathds{1} - \{B_0,B_1\} \right] \ket{\psi_{\aen|\xe}} \right\|^2 
    \leq \eta_{ anticomm }(\epsilon,\kappa, \nu, \mu)
\end{equation}

We can adapt the procedure for $x=1$. We start by reformulating Eq. \eqref{eq:Bhat_squared_comm} for the observable $\hat{B}_1$:
\begin{align*}
    (\hat{B}_1)^2 \ket{\psi_{\aen|\xe}}
&=\left(\frac{\cos^2(\chi)+\cos^2(\chi+\mu)}{\sin^2(\mu)} \mathds{1} + \frac{\cos(\chi)\cos(\chi+\mu)}{\sin^2(\mu)}  \{B_0,B_1\} \right)\ket{\psi_{\aen|\xe}}
\end{align*}
Following the same steps, we get an equivalent statement for the anticommutator averaged over $\xe = \enc(1)$:
\begin{align*}
\mathop{\mathds{E}}_{\xe=\enc(1)} \sum_{\aen} 
    \left\| \left[ 2 \cos(\mu)\mathds{1} + \{B_0,B_1\} \right] \ket{\psi_{\aen|\xe}} \right\|^2 &\leq \sin^2(\mu) \left[ 2 \cos(\mu)+\frac{1}{\cos(\chi)\cos(\chi+\mu)} \right]
    \left( \frac{\varepsilon}{c_1} + \frac{c_0+c_1}{c_1}\delta_{\mathrm{QHE}}(\kappa)\right)\\
    & \leq \eta^{'}_{ anticomm }(\epsilon,\kappa, \chi, \mu).
\end{align*}

These equations imply that the anticommutator of $B_0$ and $B_1$ is approximately equal to $2\cos(\mu)\idd$. For binary observables, this is enough to prove robust self-testing of Bob's observables. The proof of this statement can be found in \cite{Kan_STcommutation}.

\section{Computational self-test of three Pauli observables}

In this section, we present a single-prover self-testing protocol for the triplet of Pauli observables. As in previous cases, we rely on a known inequality—the Elegant Bell inequality—whose quantum bound is preserved and for which a cryptographic sum-of-squares (SOS) decomposition is well-defined. The SOS decomposition then serves as the key tool for establishing the robust self-testing result.
The result provides a complete toolkit for certifying an arbitrary single-qubit state or operation under standard cryptographic assumptions.

The Elegant Bell inequality, first introduced in \cite{Gisin2009}, is defined for two parties and binary observables, with input sets $\X \in [4]$ and $\Y \in [3]$.
The Bell operator has the form
\begin{align*}
    \mathcal{B}^{el} =  (A_0 + A_1 - A_2 - A_3)\otimes B_0 
    + (A_0 - A_1 + A_2 - A_3)\otimes B_1 
    + (A_0 - A_1 - A_2 + A_3)\otimes B_2
\end{align*}
The quantum bound $\beta_q= 4 \sqrt{3}$ self-tests a maximally entangled pair of qubits and all three Pauli observables measured by Bob~\cite{Acin_2016}. 

The nonlocal game corresponding to the elegant Bell inequality is an XOR game, therefore Theorem~\ref{thm:qbound_XOR} applies and the quantum bound of the compiled game is preserved. 
In the following section we prove that the self-testing properties of this game are also preserved by Kalai compilation.

The standard SOS decomposition for the shifted elegant Bell operator is
\begin{align}
     \beta_q \mathds{1} - \mathcal{B}^{el}=
    \frac{\sqrt{3}}{2} \sum_{i = 0}^3 P_i^2
\end{align}
with the following definition for the polynomials $P_i$:
\begin{align*}
    &P_0 = A_0 - \frac{B_0+B_1+B_2}{\sqrt{3}}, &&\qquad 
    P_1 = A_1 - \frac{B_0-B_1-B_2}{\sqrt{3}},\\
    &P_2 = A_2 - \frac{-B_0+B_1-B_2}{\sqrt{3}}, &&\qquad
    P_3 = A_3 - \frac{-B_0-B_1+B_2}{\sqrt{3}}.\\
\end{align*}
We use the standard notation $P_i = A_i - \hat{B}_i$ and repeat the procedure of taking the pseudo-expectation value of the shifted Bell operator to get the bound on the sum-of-square elements. When the compiled game reaches the score $\bar{\beta} = \beta_q - \epsilon$, the following is true
\begin{align}
    \mathop{\mathds{E}}_{\xe=\enc(x)} \sum_\aen  \left\|\hat{B}_x \ket{\psi_{\aen|\xe}} - (-1)^{\dec(\aen)}\ket{\psi_{\aen|\xe}}\right\|^2
    \leq \frac{2}{\sqrt{3}}\varepsilon + 4 \delta_{\mathrm{QHE}}(\kappa) \qquad \forall x.
\end{align}
For simplicity we study first the noiseless case, with $\delta_{\mathrm{QHE}}(\kappa) = 0$ and $\varepsilon=0$, implying that
\begin{align*}
\hat{B}_x\ket{\psi_{\aen|\xe}} - (-1)^{\dec(\aen)}\ket{\psi_{\aen|\xe}} = 0, \qquad \forall \xe:\dec(\xe) = x, \forall x.
\end{align*}
Once again, this proves that the square of the hat operator acts like identity on $\ket{\psi_{\aen|\xe}}$:
\begin{align}\label{eq:EBI_Bhatsquared}
    (\hat{B}_x)^2 \ket{\psi_{\aen|\xe}} = (-1)^{\dec(\aen)} \hat{B}_x\ket{\psi_{\aen|\xe}} = \ket{\psi_{\aen|\xe}}.
\end{align}
The square of the hat operator can be completely characterized in terms of the anti-commutators of the observables, assuming their projectivity $B_1^2 =B_2^2=B_3^2=\mathds{1}$. Let's consider $x=1$, then
\begin{align*}
    (\hat{B}_1)^2 = \mathds{1} +\frac{1}{3} \big(\{B_1,B_2\}+\{B_2,B_3\}+\{B_1,B_3\} \big).
\end{align*}
Considering all possible $x$ and Eq. \eqref{eq:EBI_Bhatsquared}, we obtain the following system of equations :
\begin{align*}
    \big(+\{B_1,B_2\}+\{B_2,B_3\}+\{B_1,B_3\} \big) \ket{\psi_{\aen|\xe}}=  0,\\
    \big(-\{B_1,B_2\}+\{B_2,B_3\}-\{B_1,B_3\} \big) \ket{\psi_{\aen|\xe}}=  0,\\
    \big(-\{B_1,B_2\}-\{B_2,B_3\}+\{B_1,B_3\} \big) \ket{\psi_{\aen|\xe}}=  0,\\
    \big(+\{B_1,B_2\}-\{B_2,B_3\}-\{B_1,B_3\} \big) \ket{\psi_{\aen|\xe}}=  0
\end{align*}
whose only solution is
\begin{align}\label{eq:EBI_anticomm}
        \{B_1,B_2\} \ket{\psi_{\aen|\xe}} = \{B_2,B_3\} \ket{\psi_{\aen|\xe}}= \{B_1,B_3\} \ket{\psi_{\aen|\xe}}=  0.
\end{align}
This completely fixes the three observales: they must be unitarely equivalent to the three Pauli's. See Appendix C of \cite{Kan_STcommutation} for a formal proof of this statement.

To make the self-testing cryptographically and experimentally meaningful it needs to be robust to noise, \emph{i.e.} we need to consider a non-zero $\varepsilon$ and $\delta_{\mathrm{QHE}}(\kappa)$.
We want to prove that all the anti-commutators are a negligible function depending on the noise $\varepsilon$ and the security parameter $\kappa$:
\begin{equation*}
     \mathop{\mathds{E}}_{\xe=\enc(0)} \sum_{\aen} 
    \left\| \{B_{y_1},B_{y_2}\}\ket{\psi_{\aen|\xe}} \right\|^2 
    \leq \eta(\epsilon,\kappa) \qquad \forall y_1 \neq y_2.
\end{equation*}

Using similar tricks as in the previous section we get the following system of equations:
\begin{align*}
    \mathop{\mathds{E}}_{\xe=\enc(0)} \sum_{\aen} 
    \left\|
    \big(+\{B_1,B_2\}+\{B_2,B_3\}+\{B_1,B_3\} \big) \ket{\psi_{\aen|\xe}} \right\|^2
    \leq 9 (1+\sqrt{3}) \left( \frac{2}{\sqrt{3}}\varepsilon+ 4 \delta_{\mathrm{QHE}}(\kappa)\right),\\
    \mathop{\mathds{E}}_{\xe=\enc(1)} \sum_{\aen} 
    \left\|
    \big(-\{B_1,B_2\}+\{B_2,B_3\}-\{B_1,B_3\} \big) \ket{\psi_{\aen|\xe}} \right\|^2
    \leq 9 (1+\sqrt{3}) \left( \frac{2}{\sqrt{3}}\varepsilon+ 4 \delta_{\mathrm{QHE}}(\kappa)\right),\\
    \mathop{\mathds{E}}_{\xe=\enc(2)} \sum_{\aen} 
    \left\|
    \big(-\{B_1,B_2\}-\{B_2,B_3\}+\{B_1,B_3\} \big) \ket{\psi_{\aen|\xe}} \right\|^2
    \leq 9 (1+\sqrt{3})\left( \frac{2}{\sqrt{3}}\varepsilon+ 4 \delta_{\mathrm{QHE}}(\kappa)\right),\\
    \mathop{\mathds{E}}_{\xe=\enc(3)} \sum_{\aen} 
    \left\|\big(+\{B_1,B_2\}-\{B_2,B_3\}-\{B_1,B_3\} \big) \ket{\psi_{\aen|\xe}} \right\|^2
    \leq 9 (1+\sqrt{3})\left( \frac{2}{\sqrt{3}}\varepsilon+ 4 \delta_{\mathrm{QHE}}(\kappa)\right).
\end{align*}
In equations above the averages are computed across various values of  $x$, but from them we can obtain the averages over the same value $\bar{x}$, albeit incurring an additional cost represented by 
$\delta_{\mathrm{QHE}}(\kappa)$ on the right-hand side. We proceed by expanding the constituent terms within the norms. Subsequently, for each equation within the system, an analogous treatment is employed, yielding expressions akin to:
\begin{align*}
    &\left\|
    \big(\{B_1,B_2\}+\{B_2,B_3\}+\{B_1,B_3\} \big) \ket{\psi_{\aen|\xe}} \right\|^2 =\\
    &=
    \left\|
    \{B_1,B_2\} \ket{\psi_{\aen|\xe}} \right\|^2
    +\left\|
    \{B_2,B_3\} \ket{\psi_{\aen|\xe}} \right\|^2
    +\left\|
   \{B_3,B_1\} \ket{\psi_{\aen|\xe}} \right\|^2+\\
   &+\bra{\psi_{\aen|\xe}} \{ \{B_1,B_2\},\{B_2,B_3\} \} \ket{\psi_{\aen|\xe}}
   +\bra{\psi_{\aen|\xe}} \{ \{B_1,B_2\},\{B_1,B_3\} \} \ket{\psi_{\aen|\xe}}
   +\bra{\psi_{\aen|\xe}} \{ \{B_1,B_3\},\{B_2,B_3\} \} \ket{\psi_{\aen|\xe}}
\end{align*}
with possibly different signs in the last line.
When we sum all of the inequalities; only the terms with the norm of the anti-commutators are going to survive, hence

\begin{align*}
    4 \mathop{\mathds{E}}_{\xe=\enc(\bar{x})} \sum_{\aen} 
    \left(
    \left\|
    \{B_1,B_2\} \ket{\psi_{\aen|\xe}} \right\|^2
    +\left\|
    \{B_2,B_3\} \ket{\psi_{\aen|\xe}} \right\|^2
    +\left\|
   \{B_3,B_1\} \ket{\psi_{\aen|\xe}} \right\|^2
    \right)
    \leq 4 (9 (1+\sqrt{3}) )
    \left( \frac{2}{\sqrt{3}}\varepsilon+ 4 \delta_{\mathrm{QHE}}(\kappa)\right).
\end{align*}
All the terms in the sum are positive, hence we obtain the desired bound for all the pairs of anti-correlators:
\begin{equation}\label{EBI_robust_anticomm}
     \mathop{\mathds{E}}_{\xe=\enc(\bar{x})} \sum_{\aen} 
    \left\| \{B_{y_1},B_{y_2}\}\ket{\psi_{\aen|\xe}} \right\|^2 
    \leq  6 (3+\sqrt{3}) \varepsilon+ 36(1+\sqrt{3}) \delta_{\mathrm{QHE}}(\kappa) \qquad \forall y_1 \neq y_2
\end{equation}
which is indeed the robust version of equation \eqref{eq:EBI_anticomm}.

\section{Open problems}
Our findings reveal that the compilation procedure introduced in~\cite{STOC:KLVY23} effectively preserves the quantum bound of XOR games. However, it remains uncertain whether this preservation extends to generic nonlocal games. As demonstrated, the compilation of Bell inequalities is feasible, yet determining whether the quantum bound preservation holds for Bell inequalities not expressible as nonlocal games poses a challenge. In both scenarios, the complexity arises from the inability to find the quantum bound using a simple correlation matrix and the duality of semi-definite programming. Resolving this likely requires using the NPA hierarchy and employing moment matrices with numerous nonobservable elements. An additional open problem pertains to the quantum behavior of compiled games involving more than two players. Modeling a single quantum prover in such instances necessitates a rigorous characterization of sequential encrypted operations. Lastly, understanding the relationship between self-testing in standard nonlocal games and their compiled counterparts proves to be an instructive avenue for future exploration.

\section*{Acknowledgements}
We thank Paul Hermouet, Dominik Leichtle, Mirjam Weilenmann, Jef Pauwels and Borivoje Daki\'{c} for insightful discussions.
MB acknowledges funding from QuantEdu France, a state aid managed by the French National Research Agency for France 2030 with the reference ANR-22-CMAS-0001. BB acknowledges support from the European Commission Horizon-CL4 program under grant agreement 101135288 for the EPIQUE project. I\v{S} and DM acknowledge funding from  PEPR integrated project EPiQ ANR-22-PETQ-0007 part of Plan France 2030. ED acknowledges funding from the European Union’s Horizon Europe research and innovation program under the grant agreement No 101114043 (QSNP).

\bibliographystyle{quantum}
\bibliography{generic,abbrev1,crypto}

\end{document}